%% file: paper.tex
\documentclass[10pt, journal]{IEEEtran}
\usepackage[utf8]{inputenc}

\usepackage{amsmath}
\usepackage{amsthm}
\usepackage{amsfonts}
\usepackage{amssymb}
\usepackage{subfigure}
\usepackage{mathptmx}
\usepackage{paralist}
\usepackage{graphicx}
\usepackage{float}
\usepackage{booktabs}
\usepackage{xspace}
\usepackage{wrapfig}
\usepackage{xfrac}

\newtheorem{proposition}{Proposition}
\newtheorem{lemma}{Lemma}
\newtheorem{definition}{Definition}

\newcommand{\sys}{Spinner\xspace}

\newcommand{\ie}{i.\,e.,\/~}

\begin{document}

\title{\sys: Scalable Graph Partitioning in the Cloud}


\author{Claudio Martella, Dionysios Logothetis, Andreas Loukas, and Georgos Siganos}

\maketitle

\begin{abstract}
Several organizations, like social networks, store and routinely analyze large
graphs as part of their daily operation. Such graphs are typically distributed
across multiple servers, and graph partitioning is critical for efficient graph
management. Existing partitioning algorithms focus on finding graph partitions
with good locality, but disregard the pragmatic challenges of integrating
partitioning into large-scale graph management systems deployed on the cloud,
such as dealing with the scale and dynamicity of the graph and the compute environment.


In this paper, we propose \emph{\sys}, a scalable and adaptive graph
partitioning algorithm based on label propagation designed on top of the Pregel
model. \sys scales to massive graphs, produces partitions with locality and
balance comparable to the state-of-the-art and efficiently adapts the
partitioning upon changes.  We describe our algorithm and
its implementation in the Pregel programming model that makes it possible to
partition billion-vertex graphs. We evaluate \sys with a variety of synthetic
and real graphs and show that it can compute partitions with quality comparable
to the state-of-the art.  In fact, by using \sys in conjunction with the Giraph
graph processing engine, we speed up different applications by a factor of $2$
relative to standard hash partitioning.  
\end{abstract}

\begin{IEEEkeywords}
Pregel, graph partitioning, graph processing, cloud computing.
\end{IEEEkeywords}


\section{Introduction}\label{sec:intro}
\input{intro}

\section{Motivation and background}\label{sec:def}
\input{background}

\section{Proposed solution}\label{sec:lpa}
\input{lpa}

\section{Pregel implementation}\label{sec:impl}
\input{impl}

\section{Evaluation}\label{sec:eval}
\input{eval}


\section{Related Work}\label{sec:related}
\input{related}


\section{Conclusions}\label{sec:concl}
\input{concl}

\bibliographystyle{shortabbrv}
\bibliography{paper}


\section*{BIOGRAPHIES}

\vspace{-50pt}
\begin{IEEEbiography}[{\includegraphics[width=1in,height=1.25in,clip,keepaspectratio]{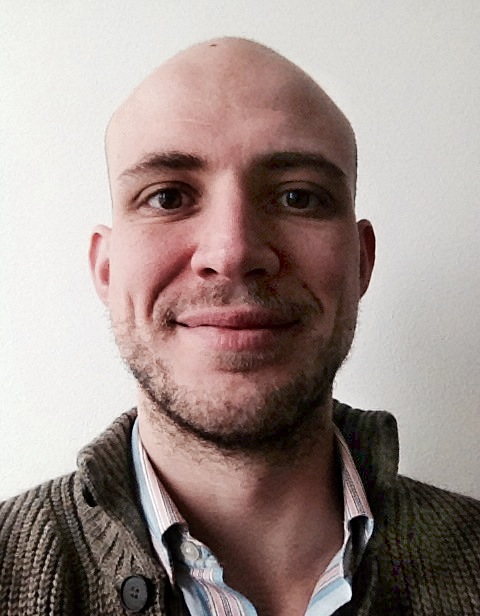}}]{Claudio Martella}
 is PhD candidate at the Computer Systems group of VU University Amsterdam, where he researches complex networks. In particular, he is interested in modelling problems with spatio-temporal networks and investigating platforms for large-scale graph processing.
He is also an active contributor to Apache Giraph and various projects in the Hadoop ecosystem.
\vspace{-50pt}
\end{IEEEbiography}

\begin{IEEEbiography}[{\includegraphics[width=1in,height=1.25in,clip,keepaspectratio]{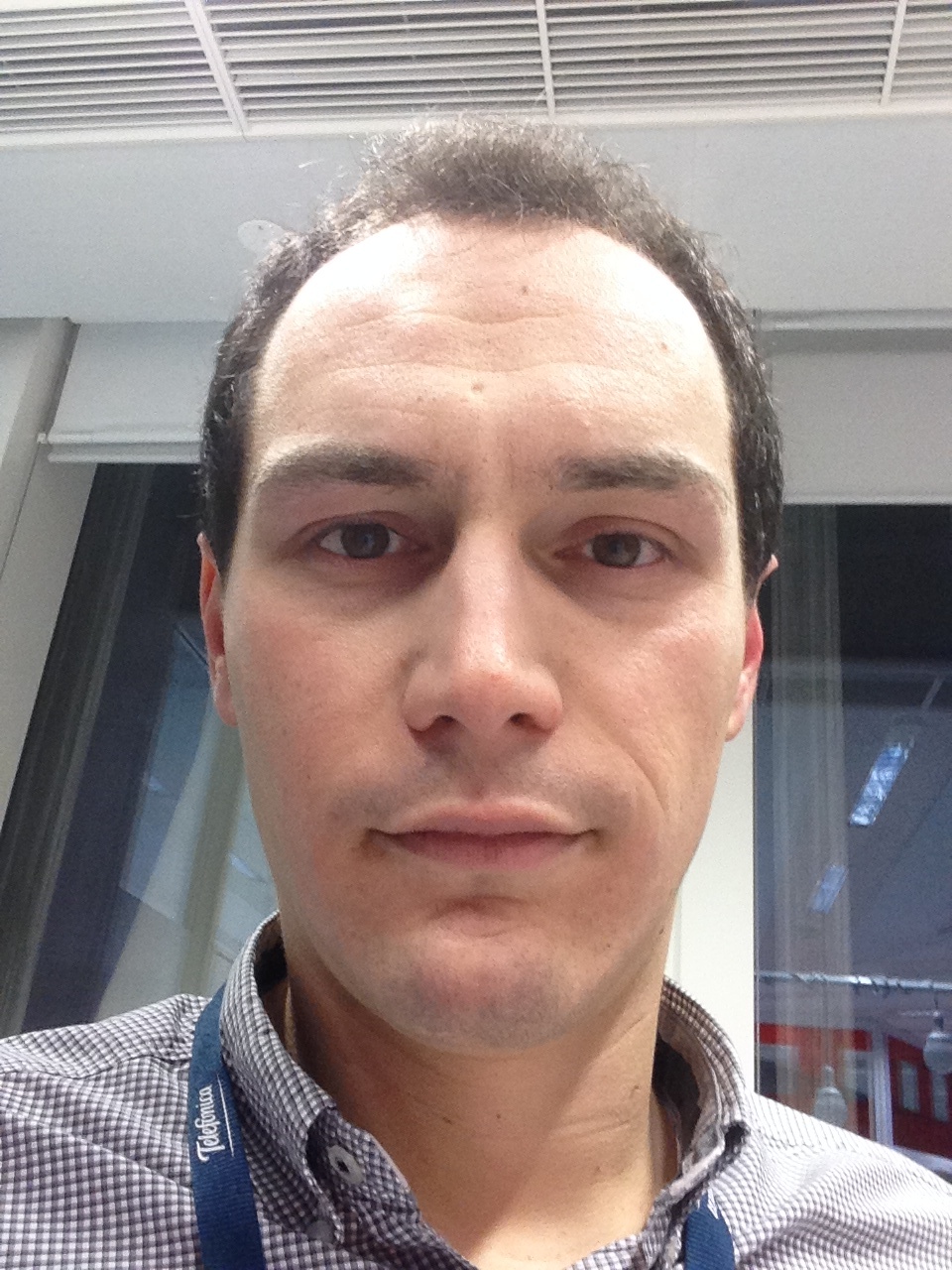}}]{Dionysios Logothetis}
 is an Associate Researcher with the Telefonica
 Research lab in Barcelona, Spain. His research interests lie in the
 areas of large scale data management with a focus on graph mining,
 cloud computing and distributed systems. He holds a PhD in Computer
 Science from the University of California, San Diego and a Diploma in
 Electrical and Computer Engineering from the National Technical
 University of Athens.
 \vspace{-50pt}
\end{IEEEbiography}

\begin{IEEEbiography}[{\includegraphics[width=1in,height=1.25in,clip,keepaspectratio]{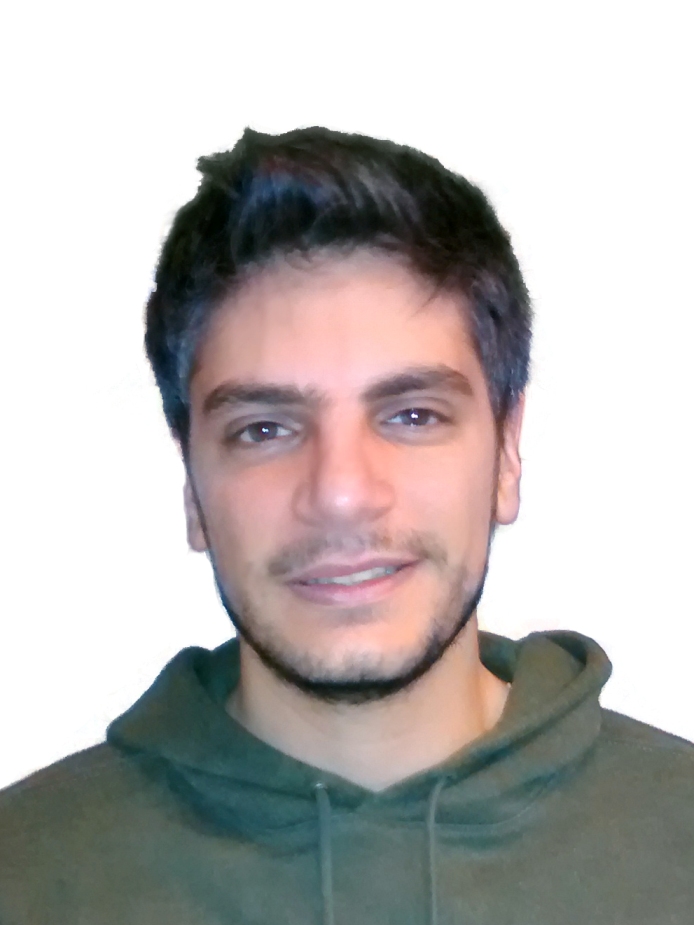}}]{Andreas Loukas}
attained a diploma in Computer Science and Engineering from Patras University, in Greece. He also pursued a PhD with the Embedded Software group at Delft University of Technology. There he gained engineering experience in working with real wireless systems, but also realized his appreciation for rigorousness. Currently, he is doing a postdoc on the same group focusing on distributed signal processing on graphs. His research interests include graph analysis and distributed algorithms.
\vspace{-50pt}
\end{IEEEbiography}

\begin{IEEEbiography}[{\includegraphics[width=1in,height=1.25in,clip,keepaspectratio]{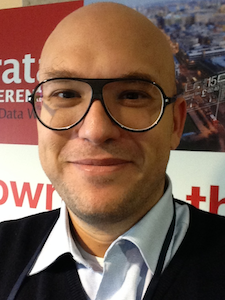}}]{Georgos Siganos}
is a Senior Scientist at Qatar Computing Research Institute working on next generation Graph Mining Architectures and Big Data Systems. Previous to this, he was a Research Scientist at Telefonica Research in Barcelona Spain focusing on Big Data and Peer to Peer Systems. He has authored more than 30 papers in journals and conferences. He received his Ph.D. from the University of California, Riverside.
\end{IEEEbiography}

\section*{APPENDIX A}
\input{appendix}


\end{document}

%% file: intro.tex

Graph partitioning is a core component for scalable and efficient graph
processing. The size of graphs available today, whether a social network, the
Web, or a protein interaction network~\cite{He2008}, requires graph management
systems~\cite{Malewicz2010, Gonzalez2012, Pujol2011, Shao2012} to distribute
the graph, typically across shared-nothing clusters or clouds.  Due to the
inherent dependencies that characterize graphs, producing partitions with good
locality is critical in minimizing communication overheads and improving system
scalability\cite{LumsdaineGHB07}. At the same time, partitions that balance
processing across the cluster can improve resource utilization and reduce
overall processing latency.

However, integrating graph partitioning in graph management systems introduces
a number of pragmatic challenges that are disregarded by existing solutions.
First, the \emph{scale} of the graphs, often reaching billions of vertices and
edges, makes it hard to provide locality and balance at the same time. For
instance, the Metis algorithm~\cite{Karypis1995, Karypis1999} produces
partitionings with good locality and balance, but doing so incurs high
computational cost and requires a global view of the graph that makes it
impractical for large graphs~\cite{Tsourakakis2014, Khayyat2013, Tian2013}.
Alternative algorithms maintain balanced partitions by introducing centralized
components to enforce strict constraints on the size of the
partitions~\cite{Ugander2013}, limiting scalability as well.  As a result,
systems often resort to lightweight solutions, such as hash partitioning,
despite the poor locality that it offers.




Second, graph management systems need to continuously adapt the partitioning of
the graph~\cite{Pujol2011}.  One one hand, graphs are naturally \emph{dynamic},
with vertices and edges constantly added and removed. Thus, the system must
regularly update the partitioning to maintain optimal application performance.
On the other hand, modern computing platforms are \emph{elastic}, allowing to
scale up and down on demand depending on workload requirements.  This implies
that the system must re-distribute the graph across a different number of
machines.
To maintain a good partitioning under these conditions, existing solutions
continuously re-execute the partitioning algorithm from scratch.  This
expensive approach prohibits frequent updates and wastes computational
resources. Furthermore, it does not consider the previous state of the
partitioned graph; a re-partitioning may compute totally different vertex
locations. This forces the system to shuffle a large fraction of the graph,
impacting performance. 

At the same time, recently emerged graph processing systems, like
Pregel~\cite{Malewicz2010} and Graphlab~\cite{Low2012}, provide platforms
customized for the design and execution of graph algorithms at scale. They
offer programming abstractions that naturally lead to distributed algorithms.
Further, their execution environments allow to scale to large graphs leveraging
clusters consisting of commodity hardware. In fact, this has resulted in
porting several graph algorithms on such systems~\cite{Quick2012, Perozzi2013,
graphlab}. However, no work has explored so far how the problem of graph
partitioning can benefit from these architectures.

In this paper, we propose to leverage such architectures for the development
of scalable graph partitioning algorithms. We introduce \emph{\sys}, a
partitioning algorithm that runs on top of Giraph~\cite{giraph}, an open source
implementation of the Pregel model~\cite{Malewicz2010}. \sys can partition
billion-vertex graphs with good locality and balance, and efficiently adapts to
changes in the graph or changes in the number of partitions.


\sys uses an iterative vertex migration heuristic based on the Label
Propagation Algorithm (LPA).  We extend LPA to a scalable, distributed
algorithm on top of the Pregel (Giraph) programming model.  As a design choice,
\sys avoids expensive centralizations of the partitioning
algorithm~\cite{Karypis1995, Ugander2013, Tsourakakis2014} that may offer
strict guarantees about partitioning at the cost of scalability. However, we
show that building only on the programming primitives of Pregel, \sys can still
provide good locality and balance, reconciling scalability with partitioning
quality.


Further, we designed \sys to efficiently adapt a partitioning upon changes to
the graph or the compute environment. By avoiding re-computations from scratch,
\sys reduces the time to update the partitioning by more than 85\% even for
large changes ($2\%$) to the graph, allowing for frequent adaptation, and
saving computation resources.  Further, the incremental adaptation prevents
graph management systems from shuffling large portions of the graph upon
changes.

Finally, we show that we can use the resulting partitionings of \sys to improve
the performance of graph management systems. In particular, we use \sys to
partition the graph in the Giraph graph analytics engine, and show that this
speeds up processing of applications by up to a factor of 2.

In this paper, we make the following contributions:

\begin{compactitem} 
	\item We introduce \sys, a scalable graph partitioning algorithm based
		on label propagation that computes k-way balanced partitions
		with good locality.
	\item We provide a scalable, open source implementation of the
		algorithm on top of the Giraph graph processing system. We
		extend the basic LPA formulation to support a decentralized
		implementation that reconciles scalability and balance
		guarantees. To the best of our knowledge, this is the first
		implementation of a partitioning algorithm on the Pregel model.
	\item We extend the core algorithm to adapt an existing partitioning
		upon changes in the graph or the or the number of partitions in
		an incremental fashion. This minimizes computation and offers
		partitioning stability.
	\item We evaluate \sys extensively, using synthetic and real graphs.
		We show that \sys scales near-linearly to billion-vertex
		graphs, adapts efficiently to dynamic changes, and
		significantly improves application performance when integrated
		into Giraph.
\end{compactitem}

The remaining of this paper is organized as follows. In Section~\ref{sec:def},
we outline \sys's design goals and give a brief overview of graph processing
systems. In Section~\ref{sec:lpa} we describe the \sys algorithm in detail.
Section~\ref{sec:impl} describes the implementation of \sys in the Pregel
model, while in Section~\ref{sec:eval} we present a thorough evaluation. In
Section~\ref{sec:related} we discuss related work, and Section~\ref{sec:concl}
concludes our study.

%% file: background.tex
To motivate our approach, in this section we first describe our design goals.
We then give a brief overview of graph processing architectures and discuss how
we can exploit such systems to build a scalable partitioning solution.

\subsection{Design goals}

At a high level, \sys must produce good partitionings but also be practical.
It must be able to manage the challenges arising from graph management systems,
such as scale and dynamism.

\textbf{Partitioning quality.} \sys must compute partitions with good
\emph{locality} and \emph{balance}. Typically, graph processing
systems~\cite{Malewicz2010, Low2012} and graph
databases~\cite{Shao2012,Wang2013} distribute vertices across machines. Because
communication in such systems coincides with graph edges, network traffic
occurs when edges cross partition boundaries. Therefore, a good partitioning
algorithm must minimize the cut edges~\footnote{The edges that cross partition
boundaries.}. Further, partitions that balance the load across machines improve
processing latency and resource utilization. For instance, in graph processing
systems the load depends on the number of messages exchanged and therefore on
the number of edges. In fact, although our approach is general, here we will
focus on balancing partitions on the number of edges they contain.
Additionally, unlike traditional community detection algorithms, we must
control the number $k$ of the partitions, as this is mandated by the underlying
compute environment, for example, the number of available machines or cores.
This objective is commonly referred to as \emph{k-way} partitioning.


\textbf{Scalability.} We must be able to apply \sys to billion-vertex graphs.
As a design principle, our algorithm must be lightweight and lend itself to a
distributed implementation. Note that providing strict guarantees about
locality and balance and, at the same time, scale to large graphs may be
conflicting requirements. Often, providing such guarantees requires expensive
coordination or global views of the graph \cite{Karypis1995}, preventing
scalability.  

In this paper, we take a different approach as we are willing to trade
partitioning quality and forego strict guarantees in favor of a more practical
algorithm that can scale.  In subsequent sections, we show how relaxing these
requirements allows a scalable distributed implementation that still produces
good partitions for all practical purposes.

\textbf{Adaptability.} \sys must be able adapt the partitioning to changes in
the graph or the underlying compute environment. For efficiency, we want to
avoid re-partitioning the graph from scratch, rather compute a new good
partitioning incrementally.  We also want to avoid large changes in the adapted
partitioning as they may force the shuffling of vertices in the underlying
graph management system.  Therefore, our algorithm should take into
consideration the most recent state of the partitioning.

\subsection{Leveraging the Pregel model}

In this paper, we propose to design a graph partitioning solution on top of
graph processing architectures such as \cite{Malewicz2010, giraph, Low2012,
Xin2013}.  These systems advocate a vertex-centric programming paradigm that
naturally forces users to express graph algorithms in a distributed manner.
Further, such systems are becoming an integral part of the software stack
deployed on cloud environments. Building on top of such a system renders our
partitioning solution practical for a wide range of compute environments.

In such models, graph algorithms are expressed through a vertex computation
function that the system computes for each vertex in the graph.  Vertex
computations are executed in parallel by the underlying runtime and can
synchronize through messages~\cite{Malewicz2010} or shared
memory~\cite{Low2012}.

\sys builds upon Pregel in particular, an abstraction that is
based on a synchronous execution model and makes it easy to scale to large
graphs. In Pregel, a program is structured as a sequence of well-defined
computation phases called \emph{supersteps}. During a superstep every vertex
executes the user-defined compute function in parallel. Vertices have an
associated state that they access from within the user-defined function, and
may communicate with each other through messages. In Pregel, messages are
delivered in a synchronous manner only at the end of a superstep.

The synchronous nature of the Pregel model avoids any expensive coordination
among vertices, allowing programs to scale near-linearly. In the following, we
will show that building a graph partitioning algorithm with this primitive
allows us to apply \sys to billion-vertex graphs.


%% file: lpa.tex
We have designed \sys based on the \emph{Label Propagation Algorithm (LPA)}, a
technique that has been used traditionally for community
detection~\cite{Barber2009}. We choose LPA as it offers a generic and well
understood framework on top of which we can build our partitioning algorithm as
an optimization problem tailored to our objectives.  In the following, we
describe how we extend the formulation of LPA to achieve the goals we set in
Section~\ref{sec:def}. We extend LPA in a way that we can execute it in a
scalable way while maintaining partitioning quality and efficiently adapt a
partitioning upon changes.


Before going into the details of the algorithm, let us introduce the necessary
notation. We define a graph as $G = \langle V, E \rangle$, where $V$ is the
set of vertices in the graph and $E$ is the set of edges such that an edge $e
\in E$ is a pair $(u, v)$ with $u, v \in V$.  
We denote by $N(v) = \{ u \colon u \in V, (u, v) \in E \}$ the neighborhood of
a vertex $v$, and by $deg(v) = |N(v)|$ the degree of $v$. In a $k$-way
partitioning, we define $L$ as a set of labels $L = \{ l_{1}, \dots, l_{k} \}$
that essentially correspond to the $k$ partitions.  $\alpha$ is the labeling
function $\alpha\colon V \to L$ such that $\alpha(v) = l_{j}$ if label $l_{j}$
is assigned to vertex $v$.  

The end goal of \sys is to assign partitions, or labels, to each vertex such
that it maximizes edge locality and partitions are balanced. 

\subsection{K-way Label Propagation}
\label{sec:LPA}

We first describe how to use basic LPA to maximize edge locality and then
extend the algorithm to achieve balanced partitions.  Initially, each vertex in
the graph is assigned a label $l_{i}$ at random, with $0 < i \leq k$.
Subsequently, every vertex iteratively propagates its label to its neighbors.
During this iterative process, a vertex acquires the label that is more
frequent among its neighbors. Specifically, every vertex $v$ assigns a
different \emph{score} for a particular label $l$ which is equal to the number
of neighbors assigned to label $l$
\begin{align}
score(v, l) = \sum_{u \in N(v)} \delta(\alpha(u), l)
\label{form:score}
\end{align}
where $\delta$ is the Kronecker delta. Vertices show preference to labels with
high score. More formally, a vertex updates its label to the label $l_{v}$ that
maximizes its score according to the update function
\begin{align}
l_{v} = \underset{l}{\arg\max} ~score(v, l)
\label{form:LPA}
\end{align}
We call such an update a \emph{migration} as it represents a logical vertex
migration between two partitions.  

In the event that multiple labels satisfy the update function, we break ties
randomly, but prefer to keep the current label if it is among them.  This
break-tie rule improves convergence speed \cite{Barber2009}, and in our
distributed implementation reduces unnecessary network communication (see
Section~\ref{sec:impl}).  The algorithm halts when no vertex updates its label.



Note that the original formulation of LPA assumes undirected graphs. However,
very often graphs are directed (e.g. the Web). Even the data models of systems
like Pregel allow directed graphs, to support algorithms that are aware of
graph directness, like PageRank.  To use LPA as is, we would need to convert
a graph to undirected. The naive approach would be to create an undirected edge
between vertices $u$ and $v$ whenever at least one directed edge exists between
vertex $u$ and $v$ in the directed graph.  

This approach, though, is agnostic to the communication patterns of the
applications running on top.  Consider the example graph in Figure
\ref{fig:conversion} that we want to partition to 3 parts.  In the undirected
graph (right), there are initially 3 cut edges.  
At this point, according to the LPA formulation, which is agnostic of the
directness of the original graph, any migration of a vertex to another
partition is as likely, and it would produce one cut edge less. 

However, if we consider the directness of the edges in the original graph,
not all migrations are equally beneficial.  In fact, either moving vertex 2 to
partition 1 or vertex 1 to partition 3 would in practice produce less cut edges
in the directed graph.  Once the graph is loaded into the system and messages
are sent across the directed edges, this latter decision results in less
communication over the network.

\begin{figure}
  \centering
    \includegraphics[width=7cm]{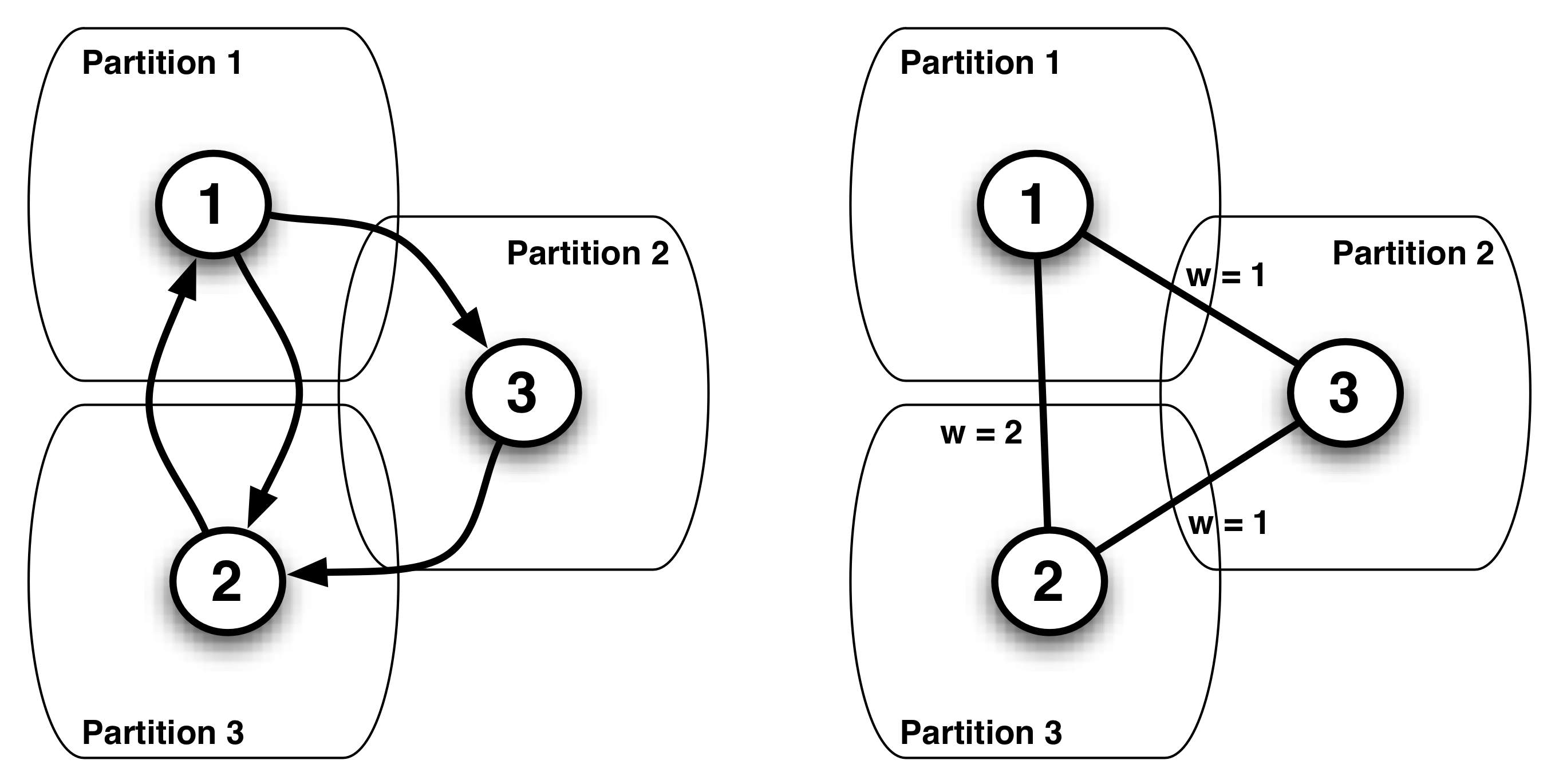}
  \caption{Conversion of a directed graph (left) to an undirected graph (right).}
  \label{fig:conversion}
  \vspace{-10px}
\end{figure}

\sys considers the number of directed edges connecting $u, v$ in the original
directed graph $D$ by introducing a weighting function $w(u,v)$ such that
\begin{align} 
   w(u, w) =
    \begin{cases}
      1, & \text{if} ~ (u, v) \in D \oplus (v, u) \in D \\
      2, & \text{if} ~ (u, v) \in D \wedge (v, u) \in D
    \end{cases}
\end{align}
where $\oplus$ is the logical XOR. We extend now the formulation in
(\ref{form:score}) to include the weighting function
\begin{align}
score'(v, l) = \sum_{u \in N(v)} w(u,v)\delta(\alpha(u), l)
\label{form:LPA2}
\end{align}
In practice, the new update function effectively counts the number of
messages exchanged locally in the system.

Notice that, so far, the formulation of LPA does not dictate in what order and
when vertices propagate their labels or how we should parallelize this process.
For instance, the propagation can occur in an asynchronous manner, with a
vertex propagating its label to its neighbors immediately after an update.  A
synchronous propagation, instead, occurs in distinct iterations with every
vertex propagating its label at the end of an iteration. In
Section~\ref{sec:impl}, we will see how we constraint the propagation to occur
in a synchronous fashion to retro-fit LPA to the Pregel model, allowing the
massive parallelization of a single LPA iteration.

\subsection{Balanced Label Propagation}
\label{sec:balanced}

Next, we will extend LPA to produce balanced partitions.  In \sys, we take a
different path from previous work \cite{Ugander2013} that balances partitions by
enforcing strict constraints on the partition sizes.  Such an approach requires
the addition to LPA of a centralized component to ensure the satisfaction of
the balance constraints across the graph. Essentially, it calculates which of
the possible migration decisions will not violate the constraints. This
component is used after each LPA iteration, potentially increasing the
algorithm overhead and limiting scalability. 

Instead, as our aim is to provide a practical and scalable solution, \sys
relaxes this constraint, only \emph{encouraging} a similar number of edges
across the different partitions. In particular, to maintain balance, we
integrate a \emph{penalty function} into the vertex score in (\ref{form:LPA2})
that penalizes migrations to partitions that are nearly full. Importantly, we
define the penalty function so that we can calculate it in a scalable manner.


In the following, we consider the case of a \emph{homogeneous} system, where
each machine has equal resources. This setup is often preferred, for instance,
in graph processing systems like Pregel, to eliminate the problem of stragglers
and improve processing latency and overall resource utilization.

We define the \emph{capacity} $C$ of a partition as the maximum number of edges
it can have so that partitions are balanced, which we set to
\begin{align}
C = c \cdot \frac{| E |}{k}
\label{form:capacity}
\end{align}
where $c>1$ is a constant, and the \emph{load} $b(l)$ of a partition $l$ as the actual number
of edges in the partition
\begin{align}
b(l) = \sum_{v \in G} deg(v) \delta(\alpha(v), l)
\end{align}

The capacity $C$ represents the constraint that \sys puts on the load of the
partitions during an LPA iteration, and for homogeneous systems it is the same
for every partition. Notice that in an ideally balanced partitioning, every
partition would contain $|E|/k$ edges. However, \sys uses parameter $c$ to let
the load of a partition exceed this ideal value.  This allows for vertices to
migrate among partitions, potentially improving locality, even if this is going
to reduce balance.

At the same time, to control the degree of unbalance, we introduce the
following penalty function that discourages the assignment of vertices to
nearly full partitions. Given a partition $l$, we define the penalty function
$\pi(l)$ as
\begin{align}
\pi(l) =  \frac{b(l)}{C}
\end{align}
The closer the current load of a partition to its capacity is, the higher the
penalty of a migration to this partition is, with the penalty value ranging
from 0 to 1. Next, we integrate the penalty function into the score
function. To do so, we first normalize
(\ref{form:LPA2}), and reformulate the score function as follows
\begin{align}
score''(v, l) = \sum_{u \in N(v)} \frac{w(u,v)\delta(\alpha(u), l)}{\sum_{u \in N(v)} w(u,v)} - \pi(l)
\label{form:LPA3}
\end{align}

This penalty function has the following desirable properties.  First, using
parameter $c$ we can control the tradeoff between partition unbalance and
convergence speed. A larger value of $c$ increases the number of migrations
allowed to each partition at each iteration. This possibly speeds up
convergence, but may increase unbalance, as more edges are allowed to be
assigned to each partition over the ideal value $|E|/k$.

Second, it allows us to compute the score function in a scalable way.  Notice
that the locality score depends on per-vertex information.  Further, computing
the introduced penalty function only requires to calculate $b(l)$. This is an
aggregate across all vertices that are assigned label $l$. As we describe in
Section~\ref{sec:impl}, we can leverage the Giraph aggregation primitives to
compute $b(l)$ for all possible labels in a scalable manner.  As we show later,
the introduction of this simple penalty function is enough to produce
partitions with balance comparable to the state-of-the-art. 

\subsection{Convergence and halting}

Although proving the convergence properties of LPA is a hard problem in the general
case~\cite{Barber2009}, our additional emphasis on partition balance, namely that a vertex can migrate to improve balance despite a decrease in locality, allow us to provide some analytical guarantees about \sys's convergence and partitioning quality. 
One of the difficulties in proving convergence is that approaches based on LPA sometimes reach a limit cycle where the partitioning fluctuates between the same states. Nevertheless, we can analyze it from a stochastic perspective.
By using classic results from the theory of products of stochastic matrices~\cite{tsitsiklis1984problems,touri2012}, in the following we show that: (i) under sufficient conditions on the connectivity of the underlying graph, \sys converges exponentially-fast to an even balancing, and (ii) in the general case \sys convergences, but no guarantees are provided about the partitioning quality. 
For practical purposes, in this section we also provide a heuristic for deciding when to halt the execution of the algorithm.




\textbf{Model.} We will abstract from the underlying graph by encoding the system state in a $k$-dimensional load vector $x = [B(l_1), B(l_2), \ldots B(l_k)] $. At each iteration $t$, Spinner moves a portion of the load of each partition to each of the other $k-1$ partitions. This can be captured by a partition graph $P_t = (L, Q_t)$, having one vertex for each label and an edge $(l_i,l_j) \in Q_t$ if and only if the portion $[X_t]_{ij}$ of load which migrates from partition $i$ to $j$ is larger than zero. 
Supposing that $x_0$ is the starting state, the state $x_t$ during iteration $t$ is given by
\begin{align}
	x_t = X_t\, X_{t-1} \cdots X_{1} \, x_{0} = X_{t:1}\, x_0. 
\end{align}
We next show that, when the partition graph is $B$-connected, \sys converges exponentially-fast to an even balancing.
\begin{definition}[$B$-connectivity]
A sequence of graphs $\{P_{t>0}\}$ is $B$-connected if, for all $t>0$, each union graph $P_{Bt:B(t+1)-1} = (L, \, Q_{Bt} \cup Q_{Bt + 1} \cup \cdots \cup Q_{B(t+1)-1})$ is strongly connected.  
\label{def:bconnectivity}
\end{definition}
Simply put, $B$-connectivity asserts that each partition exchanges load with every other partition periodically. We can now state our first result:
\begin{proposition}
	If the partition graph $\{P_{t>0}\}$ is $B$-connected, one can always find constants $\mu \in (0,1)$ and $q \in \mathbf{R}^{+}$, for which Spinner converges exponentially $\|x_t - x^{\star} \|_{\infty} / \|x_0\|_{\infty} \leq q \mu^{t-1}$ to an even balancing $x^{\star} = [C\ C\ \ldots\ C]^\top\hspace{-1mm}$, with $C = |E|/k$. 
	\label{prop:bconnectivity}
\end{proposition}

Proposition~\ref{prop:general} asserts that, when $\{P_{t>0}\}$ is not $B$-connected, Spinner also converges---though we cannot state anything about the quality of the achieved partitioning. 
\begin{proposition}
	\sys converges in bounded time.
	\label{prop:general}
\end{proposition}

Due to lack of space, we refer to \cite{extended} for the derivation of the proofs of Propositions~\ref{prop:bconnectivity} and~\ref{prop:general}.

From a practical perspective, even if limit cycles are avoided often it is not worth spending compute cycles to achieve full convergence.  Typically, most of the
improvement in the partitioning quality occurs during the first iterations, and
the improvement per iteration drops quickly after that. We validate this with
real graphs in Section~\ref{sec:quality}.

In LPA convergence is detected by the absence of vertices changing label, referred to as the halting condition. A number of strategies have been proposed to guarantee the halting of LPA in synchronous systems (such as Pregel). These strategies are either based on heuristics for tie breaking and halting, or on the order in which vertices are updated~\cite{Yang2012a}. However, the heuristics are tailored to LPA's score function, which maximizes only locality. Instead, our score function does not maximize only locality, but also partition balance, rendering these strategies unsuitable. Hence, in \sys we use a heuristic that tracks how the quality of partitioning improves across the entire graph.
%
%

At a given iteration, we define the \emph{score} of the partitioning for graph $G$ as the sum of the current scores of each vertex
\begin{align}
score(G) = \sum_{v \in G} score''(v, \alpha(l_{v}))
\end{align}
As vertices try to optimize their individual scores by making local decisions,
this aggregate score gradually improves as well. We consider a partitioning to
be in a \emph{steady state}, when the score of the graph is not improved more
than a given threshold $\epsilon$ for more than $w$ consecutive iterations. The algorithm
halts when a steady state is reached.  Through $\epsilon$ we can control the
trade-off between the cost of executing the algorithm for more iterations and
the improvement obtained by the score function. At the same time, with $w$ it
is possible to impose a stricter requirement on stability; with a larger $w$,
we require more iterations with no significant improvement until we accept to
halt.

Note that this condition, commonly used by iterative hill-climbing optimization
algorithms, does not guarantee halting at the optimal solution. However, as we
present in Section \ref{sec:incr}, \sys periodically restarts the partitioning
algorithm to adapt to changes to the graph or the compute environment. This
natural need to adapt gives \sys the opportunity to jump out of local optima.

\subsection{Incremental Label Propagation}
\label{sec:incr}

As edges and vertices are added and removed over time, the computed
partitioning becomes outdated, degrading the global score. Upon such changes,
we want to update the partitioning to reflect the new topology without
repartitioning from scratch. Ideally, since the graph changes affect local
areas of the graph, we want to update the latest stable partitioning only in
the portions that are affected by the graph changes. 

Due to its local and iterative nature, LPA lends itself to incremental
computation.  Intuitively, the effect of the graph changes is to ``push'' the
current steady state away from the local optimum it converged to, towards a
state with lower global score. To handle this change, we restart the
iterations, letting the algorithm search for a new local optimum.  In the event
we have new vertices in the graph,  we initially assign them to the least
loaded partition, to ensure we do not violate the balance constraint.
Subsequently, vertices evaluate their new local score, possibly deciding to
migrate to another partition. The algorithm continues as described previously.

Upon graph changes, there are two possible strategies in restarting vertex
migration. The first strategy restarts migrations only for the vertices
affected by the graph changes, for instance, vertices adjacent to a newly added
edge. This strategy minimizes the amount of computation to adapt. The second
strategy, allows every vertex in the graph, even if it is not affected by a
change, to participate in vertex migration process. This latter strategy incurs
higher computation overhead, but increases the likelihood that the algorithm
jumps out of a local optimum. We have found that this computational overhead
does not affect the total running time of the algorithm significantly and,
therefore, opt for this latter strategy favoring better partitioning quality.

Note that the number of iterations required to converge to a new steady state
depends on the number of graph changes and the last state. Clearly, not every
graph change will have the same effect. Sometimes, no iteration may be
necessary at all. In fact, certain changes may not affect any vertex to the
point that the score of a different label is higher than the current one. As no
migration is caused, the state remains stable. On the other hand, other changes
may cause more migrations due to the disruption of certain weak local
equilibriums. In this sense, the algorithm behaves as a hill-climbing
optimization algorithm. As we will show in Section~\ref{sec:eval:dynamic}, even
upon a large number of changes, \sys saves a large fraction of the time of a
re-partitioning from scratch.

\subsection{Elastic Label Propagation}
\label{sec:elastic}
During the lifecycle of a graph, a system may need to re-distribute the graph
across the compute cluster.  For instance, physical machines may be
characterized by a maximum capacity in the number of vertices or edges they can
hold due to resources limitations, such as the available main memory. As the
graph grows and the load of the partitions reaches this maximum capacity, the
system may need to scale up with the addition of more machines, requiring to
re-distribute the graph. Alternatively, we may perform such re-distribution
just to increase the degree of parallelization and improve performance.
Conversely, if the graph shrinks or the number of available machines decreases,
we need to remove a number of partitions and, again, redistribute the graph. 

In these scenarios, we want the algorithm to adapt to the new number of
partitions without repartitioning the graph from scratch. 
\sys achieves this in the following way. Upon a change in the number of
partition, \sys lets each vertex decide independently whether it should migrate
using  a probabilistic approach. In the case we want to add $n$ new partitions
to the system, each vertex picks one of the new partitions randomly and
migrates to it with a probability $p$ such that
\begin{align}
p = \frac{n}{k+n}
\end{align}
In the case we want to remove $n$ partitions, all the vertices assigned to
those partitions migrate to one of the remaining ones. Each vertex chooses
uniformly at random the partition to migrate to. 

In both cases, after the vertices have migrated, we restart the algorithm to
adapt the partitioning to the new assignments. As in the case of incremental
LPA, the number of iterations required to converge to a new steady state
depends on factors, such as the graph size, and the number of partitions added
or removed.

By introducing these random migrations upon a change, this strategy clearly
disrupts the current partitioning, degrading the global score. However, it has
a number of interesting characteristics. First, it remains a decentralized and
lightweight heuristic as each vertex makes a decision to migrate independently.
Second, by choosing randomly, the partitions remain fairly balanced even after
a change in the number of partitions.  Third, this random change from the
current state of the optimization problem may allow the solution to jump
out of a local optimum. 

Note that, if the number $n$ of new partitions is large, the cost of adapting
the partitioning may be quite large, due to a large number of random
migrations. However, in practice, the frequency with which partitions are added
or removed is low compared, for example, to the number of times a partitioning
is updated due to changes in the graph itself. Furthermore, although vertices
are shuffled around, the locality of those vertices that do not migrate is not
completely destroyed, such as if the partitioning was performed from scratch.
The adaptation of the partitioning to the new number of partitions will
naturally take advantage of the late state of the partitioning.

%% file: impl.tex
We implemented Spinner in Apache Giraph \cite{giraph} and open sourced the
code\footnote{http://grafos.ml}. Giraph is an open source project with a Java
implementation of the Pregel model. Giraph is a batch graph processing system
that runs on Hadoop \cite{hadoop}, and can run computations on graphs with
hundreds of billions of edges across clusters consisting of commodity machines. 

In this section, we describe the implementation details of Spinner.  We show
how we extend the LPA formulation to leverage the synchronous vertex-centric
programming model of a system like Giraph. We implemented a distributed
algorithm with no centralized component that scales and, at the same time,
achieves good partitioning quality.


\subsection{Vertex-centric partitioning}

\begin{figure} 
\centering 
\includegraphics[width=7cm]{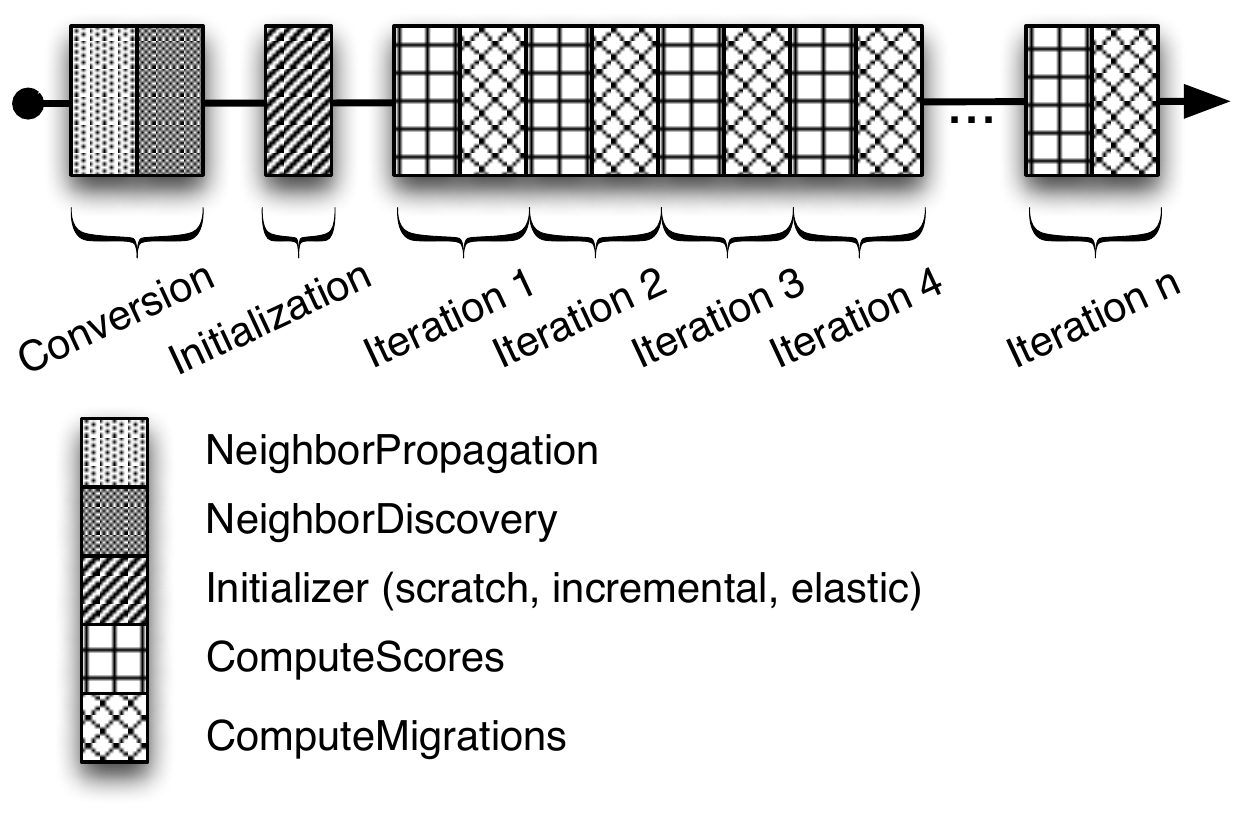}
\caption{Organization of the algorithm in multiple phases, each implemented by
one or multiple steps (block). Each algorithm step is implemented as a Pregel
superstep.}
\label{fig:phases} 
\vspace{-10px}
\end{figure}

At a high-level, the algorithm is organized in three phases, depicted in
Figure \ref{fig:phases}.  In the first phase, since LPA assumes an undirected
graph, if directed, \sys converts it to a weighted undirected form  as
described in Section \ref{sec:LPA}.  In the second phase, \sys initializes the
partition of each vertex depending on whether it partitions the graph from
scratch or it is adapting an existing partitioning.

Subsequently, the third phase that implements the main LPA iterative migration
process starts. In this phase, \sys iteratively executes two different steps.
In the first step, each vertex computes the label that maximizes its local
score function. In the second step, vertices decide whether to migrate by
changing label or defer migration. At the end of each iteration \sys evaluates
the halting condition to decide whether to continue or stop computation. 

We implement each of these phases as a series of Giraph supersteps. In the
following subsections, we describe each phase in detail.

\subsubsection{Graph conversion and initialization}

We implement the first phase, the graph conversion, as two Giraph supersteps.
Note that the Giraph data model is a distributed directed graph, where every
vertex is aware of its outgoing edges but not of the incoming ones. For this
reason, in the first superstep, each vertex sends its ID as a message to its
neighbors. We call this step \emph{NeighborPropagation}.

During the second superstep, a vertex receives a message from every other
vertex that has an edge to it.  For each received message, a vertex checks
whether an outgoing edge towards the other endpoint is already present. If this
is the case, the vertex sets the weight of the associated edge to 2.
Otherwise, it creates an edge pointing to the other vertex with a weight of 1.
We call this step \emph{NeighborDiscovery}.  

After this, \sys executes the second phase, assigning partitions to each
vertex. We call this step \emph{Initialization}, and it corresponds to a single
Giraph superstep. In the case \sys partitions the graph from scratch, each
vertex chooses a random partition.  We will consider the case of adapting a
partitioning in following sections. At this point, the LPA computation starts
on the undirected graph.

\subsubsection{Local computation of labels}

The vertex-centric programming model of Pregel lends itself to the
implementation of LPA.  During an LPA iteration, each vertex computes the label
that maximizes its local score based on the load of each partition and the
labels of its neighbors. Each vertex stores the label of a neighbor in the
value of the edge that connects them. When a vertex changes label, it informs
its neighbors of the new label through a message. Upon reception of the
message, neighboring vertices update the corresponding edge value with the new
label.

We implement a single LPA iteration as two successive Giraph supersteps that we
call \emph{ComputeScores} and \emph{ComputeMigrations}. In the first step, the
\emph{ComputeScores}, a vertex finds the label that maximizes its score.  In
more detail, during this step each vertex performs the following operations:
(i) it iterates over the messages, if any, and updates the edge values with the
new partitions of each neighbors, (ii) it iterates over its edges and computes
the frequency of each label across its neighborhood, (iii) it computes the
label the maximizes its local score, (iv) if a label with a higher score than
the current one is found, the vertex is flagged as \emph{candidate} to change
label in the next step.


\subsubsection{Decentralized migration decisions}
\label{sec:migrations}

Implementing our algorithm in a synchronous model introduces additional
complexity. If we let every candidate vertex change label to maximize its local
score during the \emph{ComputeScores} step, we could obtain a partition that is
unbalanced and violates the maximum capacities restriction. Intuitively,
imagine that at a given time a partition was less loaded than the others. This
partition could be potentially attractive to many vertices as the penalty
function would favor that partition. As each vertex computes its score
independently based on the partition loads computed \emph{at the beginning} of
the iteration, many vertices could choose independently that same partition
label. To avoid this condition, we introduce an additional step that we call
\emph{Compute Migrations}, that serves the purpose of maintaining balance.  

To avoid exceeding partition capacities, vertices would need to coordinate after
they have decided the label that maximizes their score during the
\emph{ComputeScores} step. However, as we aim for a decentralized and
lightweight solution, we opt for a probabilistic approach: a candidate vertex
changes to a label with a probability that depends (i) on the remaining capacity
of the corresponding partition and (ii) the total number of vertices that are
candidates to change to the specific label.

More specifically, suppose that at iteration $i$ partition $l$ has a remaining
capacity $r(l)$ such that
 \begin{align}
 r(l) = C - b(l)
 \end{align}
Suppose that $M(l)$ is the set of candidate vertices that want to change to
label $l$, which is determined during the \emph{ComputeScores} step of the
iteration. We define the load of $M(l)$ as
\begin{align}
\label{form:migration}
m(l) = \sum_{v\in M(l)} deg(v)
\end{align} 
This is the total load in edges that vertices in $M(l)$ would carry to
partition $l$ if they all migrated.  Since $m(l)$ might be higher than the
remaining capacity, in the second step of the iteration, we execute each
candidate vertex that wants to change label, and we only let it change with a
probability $p$ such that
\begin{align}
\label{form:migration}
p = \frac{r(l)}{m(l)}
\end{align} 
Upon change, each vertex updates the capacities of the current partition and the
new partition, and it updates the global score through the associated counter.
It also sends a message to all its neighbors, with its new label. At this point,
after all vertices have changed label, the halting condition can be
evaluated based on the global score. 

This strategy has the advantage of requiring neither centralized nor direct
coordination between vertices. Vertices can independently decide to change label
or retry in the next iteration based on local information. Moreover, it is
lightweight, as probabilities can be computed on each worker at the beginning of
the step. Because this is a probabilistic approach, it is possible that the load
of a partition exceeds the remaining capacity. Nevertheless, the probability is
bounded and decreases exponentially with the number of migrating vertices and
super-exponentially with the inverse of the maximum degree.
%
%
\begin{proposition}
	The probability that at iteration $i+1$ the load $b_{i+1}(l)$ exceeds the capacity by a factor of $\epsilon\, r_i(l)$ is 
\begin{align}
	\mathbf{Pr}(b_{i+1}(l) \geq C + \epsilon \, r_i(l)) &\leq e^{- 2\, |M(l)|\, \Phi(\epsilon)}, 	
\end{align}
where $\Phi(\epsilon) = \left(\frac{\epsilon \, r_i(l)}{\Delta - \delta}\right)^{\hspace*{-1mm}2}$ and $\delta$, $\Delta$ is the minimum and maximum degree of the vertices in $M(l)$, respectively. 
\label{prop:tail_inequality}
\end{proposition}
Due to lack of space, we refer to \cite{extended} for the derivation of the proof.

We can conclude that with high probability at each iteration \sys does not
violate the partition capacity. To give an example, consider that $|M(l)| = 200$
vertices with minimum and maximum degree $\delta =1$ and $\Delta = 500$,
respectively, want to migrate to partition $l$. The probability that, after the
migration, the load $b_{i+1}(l)$ exceeds $1.2 \, r_i(l) + b_i(l) = C + 0.2\,
r_i(l)$ is smaller than 0.2, where the probability that it exceeds $C + 0.4\,
r_i(l)$ is smaller than 0.0016. Note that, being a upper bound, this is a
pessimistic estimate. In Section \ref{sec:eval:balance} we show experimentally
that unbalance is in practice much lower.


\subsubsection{Asynchronous per-worker computation}
\label{sec:async}

Although the introduction of the \emph{ComputeMigrations} step helps maintain
balance by preventing excessive vertices from acquiring the same label, it
depends on synchronous updates of the partition loads.  The probabilistic
migration decision described in \ref{sec:migrations} is based on the partition
loads calculated during the previous superstep, and ignores any migrations
decision performed during the currently executed superstep. Consequently, a less
loaded partition will be attractive to many vertices, but only a few of them
will be allowed to migrate to it, delaying the migration decision of the
remaining ones until the next supersteps.

In general, the order in which vertices update their label impacts convergence
speed. While asynchronous graph processing systems allow more flexibility in
scheduling of vertex updates, the synchronous nature of the Pregel model does
not allow any control on the order of vertex computation.

However, \sys leverages features of the Giraph API to emulate an asynchronous
computation without the need of a purely asynchronous model. Specifically, \sys
treats each iteration computation within the same physical machine worker of a
cluster as an asynchronous computation.  During an iteration, each worker
maintains local values for the partition loads that are shared across all
vertices in the same worker. When a vertex is evaluated in the
\emph{ComputeScores} step and it becomes a candidate to change to a label, it
updates the local values of the load of the corresponding partition
asynchronously. As an effect, subsequent vertex computations in the same
iteration and on the same worker use more up-to-date partition load
information.  Note that every vertex still has to be evaluated in the
\emph{ComputeMigrations} step for consistency among workers.  


This approach overall speeds up convergence, and does not hurt the scalability
properties of the Pregel model. In fact, the \sys implementation leverages a
feature supported by the Giraph programming interface that allows data sharing
and computations on a per-worker basis.  The information shared within the same
worker is a set of counters for each partition and therefore does not add to
the memory overhead. Furthermore, it still does not require any coordination
across workers. 

\subsubsection{Management of partition loads and counters}
\label{sec:aggregators}

Spinner relies on a number of counters to execute the partitioning: the global
score, the partition loads $b(l)$, and the migration counters $m(l)$. The
Pregel model supports global computation of commutative and associative
operations through \emph{aggregators}. During each superstep, vertices can
aggregate values into named aggregators, and they can access the value
aggregated during the previous superstep. In Pregel, each aggregator is
computed in parallel by each worker for the aggregations performed by the
assigned vertices, and a master worker aggregates these values at the end of
the superstep. Giraph implements sharded aggregators, where the duty of the
master worker for each aggregator is delegated to a different worker. This
architecture allows for scalability, through a fair distribution of load and
parallel communication of partial aggregations. To exploit this feature, \sys
implements each counter through a different aggregator, making the management
of counters scalable.

\subsection{Incremental and elastic repartitioning}

To support incremental and elastic repartitioning, \sys restarts the
computation from the previous state of the partitioning. Regarding the
implementation, the difference lies in the execution of the second phase of the
algorithm, when vertices are labeled.  In the case a graph is repartitioned due
to changes to the graph, the vertices that have already been partitioned are
loaded and labeled as previously. Any new vertices are labeled randomly. In the
case a graph is repartitioned due to changes to the number of partitions, the
vertices are loaded and labeled as previously, and they are re-labeled to a new
partition only according to the approach described in Section
\ref{sec:elastic}.

%% file: eval.tex
\begin{table*}[t!]
\centering
\begin{tabular}{@{}lcllcllcllcllcll@{}}\toprule
& \phantom {} & \multicolumn{2}{c}{\textbf{Twitter k=2}} & \phantom{} & \multicolumn{2}{c}{\textbf{Twitter k=4}}  & \phantom{} & \multicolumn{2}{c}{\textbf{Twitter k=8}} & \phantom{} & \multicolumn{2}{c}{\textbf{Twitter k=16}} & \phantom{} & \multicolumn{2}{c}{\textbf{Twitter k=32}} \\
\cmidrule{3-4} \cmidrule{6-7} \cmidrule{9-10} \cmidrule{12-13} \cmidrule{15-16} 
Approach && $\phi$ & $\rho$ && $\phi$ & $\rho$ && $\phi$ & $\rho$  && $\phi$ & $\rho$ && $\phi$ & $\rho$ \\ \midrule
Wang et al. \cite{Wang2013} && $0.61$ & $1.30$ && $0.36$ & $1.63$ && $0.23$ & $2.19$ && $0.15$ & $2.63$ && $0.11$ & $1.87$ \\
Stanton et al. \cite{Stanton2012} && $0.66$ & $1.04$ && $0.45$ & $1.07$ && $0.34$ & $1.10$ && $0.24$ & $1.13$ && $0.20$ & $1.15$ \\
Fennel \cite{Tsourakakis2014} && $0.93$ & $1.10$ && $0.71$ & $1.10$ && $0.52$ & $1.10$ && $0.41$ & $1.10$ && $0.33$ & $1.10$  \\
Metis \cite{Karypis1995} && $0.88$ & $1.02$ && $0.76$ & $1.03$ && $0.64$ & $1.03$ && $0.46$ & $1.03$ && $0.37$ & $1.03$ \\
\textbf{\sys} && $0.85$ & $1.05$ && $0.69$ & $1.02$ && $0.51$ & $1.05$ && $0.39$ & $1.04$ && $0.31$ & $1.04$ \\
\bottomrule
\end{tabular}
\caption{Comparison with state-of-the-art approaches. \sys outperforms or compares to the stream-based approaches, and is only slightly outperformed by sequential Metis. Notice that because Wang et al. balances on the number of vertices, not edges, it produces partitionings with high values of $\rho$.}
\label{tab:comparison}
\vspace{0px}
\end{table*}
In this section, we assess \sys's ability to produce good partitions on large
graphs. Specifically, we evaluate the partitioning quality in terms of locality
and balance and use \sys to partition billion-vertex graphs. Furthermore, we
evaluate \sys's ability to support frequent adaptation in dynamic cloud
environments. Here, we measure the efficiency of \sys in adapting to changes in
the underlying graph and compute resources. Furthermore, we utilize the
partitioning computed by \sys with the Giraph graph processing engine and
measure the impact on the performance of real analytical applications.

For our experiments, we use a variety of synthetic as well as real-world
large-scale graphs.  Table \ref{tab:datasets} summarizes the real datasets we
used.  We run our evaluation on different Hadoop clusters. We describe the
details of each setup in the following sections.

%
\begin{table}
	\centering
	\begin{tabular}{l r r r r}
		\textbf{Name} & \textbf{|V|} & \textbf{|E|} & \textbf{Directed} & \textbf{Source} \\
		\hline
		LiveJournal (LJ) & 4.8M & 69M & Yes & \cite{Backstrom2006} \\
		Tuenti (TU) & 12M & 685M  & No & \cite{tuenti} \\
		Google+ (G+) & 29M & 462M & Yes & \cite{Gong2012} \\
		Twitter (TW) & 40M & 1.5B & Yes & \cite{Kwak2010} \\
		Friendster (FR) & 66M & 1.8B & No & \cite{Yang2012a} \\
		Yahoo! (Y!) & 1.4B & 6.6B & Yes & \cite{webscope}
	\end{tabular}
	\caption{Description of the real-world datasets used for the evaluation.}
	\label{tab:datasets}
	\vspace{-10px}
\end{table}
\subsection{Partitioning quality}
\label{sec:quality}

\begin{figure*}[t]
	\centering
	\subfigure[]{
	        \includegraphics[width=0.47\linewidth]{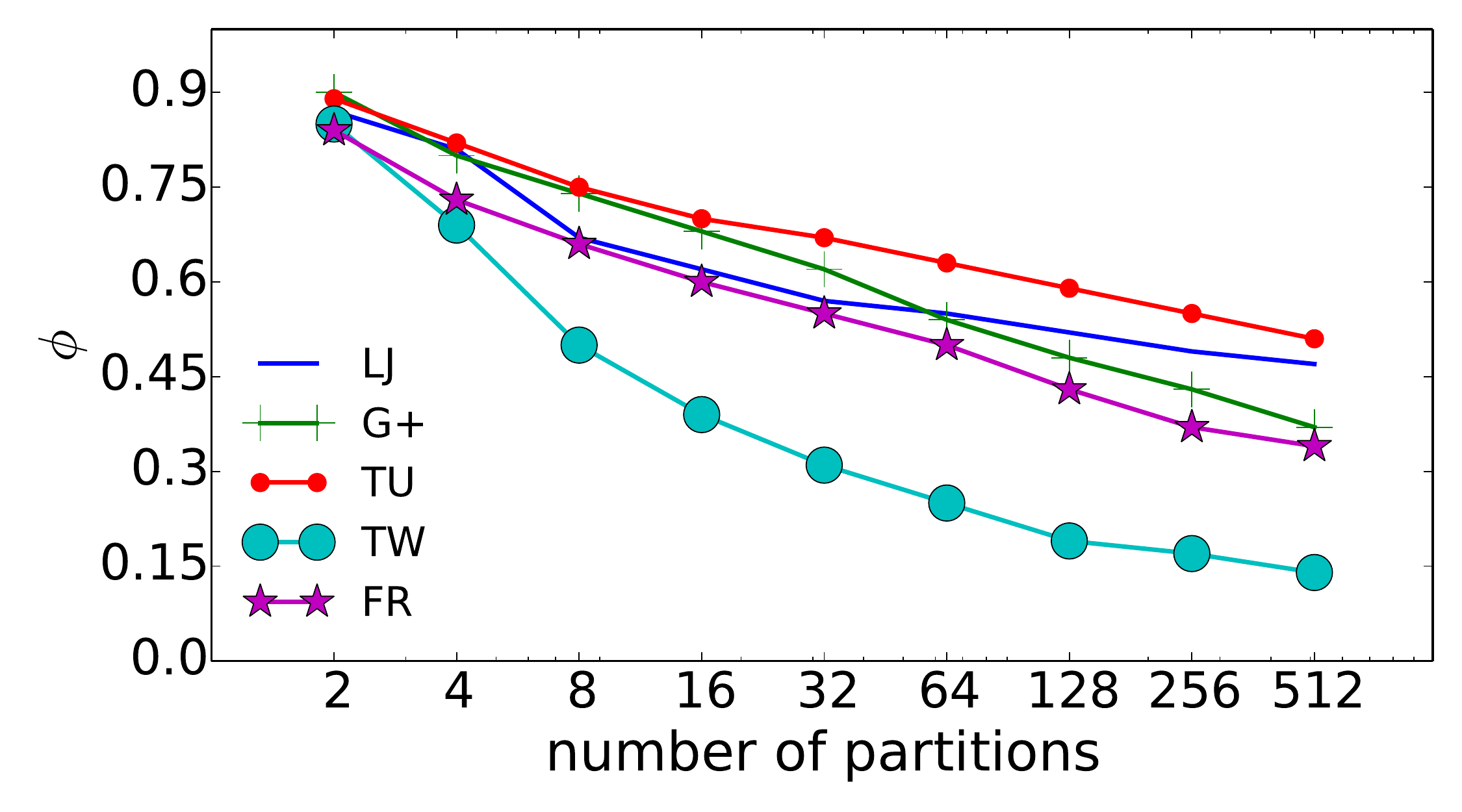} 
		\label{fig:quality}
	}
	\subfigure[]{
	        \includegraphics[width=0.47\linewidth]{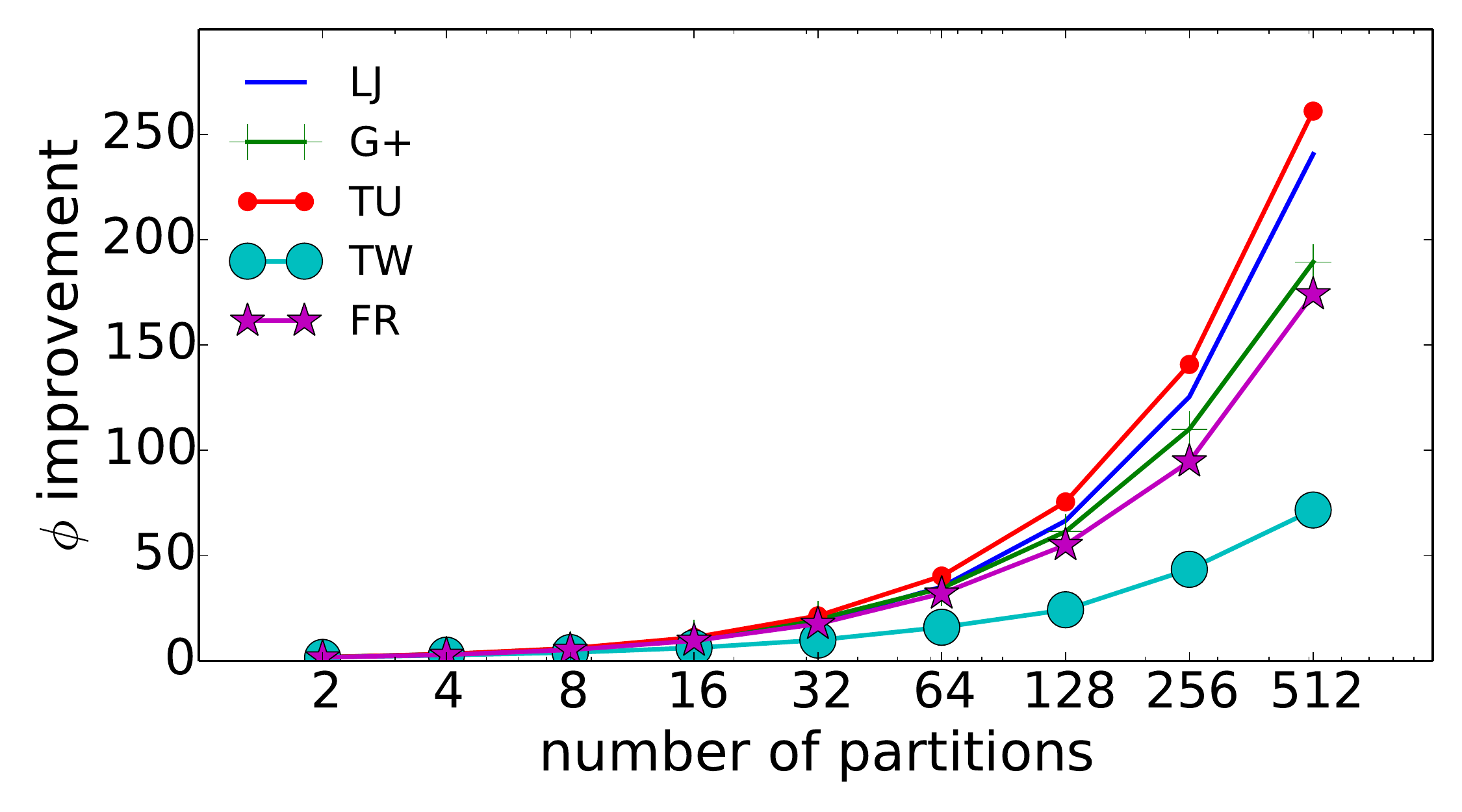}
		\label{fig:relative_perf}
	}
	\vspace{-10px}
	\caption{(a) Partitioning locality on real graphs as a function of the
	number of partitions. X-axis is in log scale. (b) Improvement in the
	locality compared to hash partitioning. X-axis is in log scale.} 
	\label{fig:partqual}
\end{figure*}

In our first set of experiments, we measure the quality of the partitions
that \sys can compute in terms of locality and balance.  We measure locality with the
\emph{ratio of local edges} $\phi$ and balance with the \emph{maximum
normalized load} $\rho $, defined as
\begin{align}
\phi= \frac{\#~local~edges}{|E|},\;
\rho = \frac{maximum~load}{\frac{|E|}{k}}
\end{align}
where $k$ is the number of partitions, $\#~local~edges$ represents the number
of edges that connect two vertices assigned to the same partition, and
$maximum~load$ represents the number of edges assigned to the most loaded
partition. The maximum normalized load metric is typically used to measure
unbalance and represents the percentage-wise difference of the most loaded
partition from a perfectly balanced partition.

First, we observe that the ability to maintain local edges depends on the number
of partitions.  Intuitively, the more partitions, the harder it is to maintain
locality.  In this experiment, we vary the number of partitions and measure
locality and balance for different graphs. For this and the remaining
experiments, we set the algorithm parameters as follows: additional capacity $c
= 1.05$, and halting thresholds $\epsilon = 0.001$ and $w= 5$.  We run this
experiment on a Hadoop cluster of 92 machines with 64GB RAM and 8 compute cores
each.


In Figure \ref{fig:quality}, we show that \sys is able to produce partitions
with high locality for all the graphs also for large numbers of partitions. With
respect to balance, \sys calculates fairly balanced partitions. In
Table~\ref{tab:balance} we show the average value of the maximum normalized load
for each graph. For example, a $\rho$ value of 1.059 for the Twitter graphs
means that no partition exceeds the ideal size by more than 5.9\% edges.

\begin{table}
	\centering
	\begin{tabular}{l  r r r r r}
		\toprule
		Graph & LJ & G+ & TU & TW & FR \\
		\hline
		$\rho$ & 1.053 & 1.042 & 1.052  & 1.059 & 1.047 \\
		\bottomrule
	\end{tabular}
	\caption{Partitioning balance. The table shows the average $\rho$ for
	the different graphs.} 
	\label{tab:balance}
	\vspace{-13px}
\end{table}


To give perspective on the quality of the partitions that \sys computes,
Figure~\ref{fig:relative_perf} shows the improvement in the percentage of local
edges compared to hash partitioning. We perform this comparison for the same set
of graphs. Notice that for 512 partitions \sys increases locality by up to 250
times.

In Table \ref{tab:comparison}, we compare \sys with state-of-the-art
approaches. Recall that our primary goal for \sys is to design a scalable
algorithm for the Pregel model that is practical in maintaining the resulting
partitioning, and that is \emph{comparable} to the state-of-the-art in terms of
locality and balance.  Indeed, \sys computes partitions with locality that is
within 2-12\% of the best approach, typically Metis, and balance that is within
1-3\% of the best approach. In cases \sys performs slightly worse than Fennel
with respect to $\phi$, it performs better with respect to $\rho$.  These two
metrics are connected as the most loaded partition will be the result of
migrations to increase locality.

To describe in more detail the behavior of our algorithm, in Figure
\ref{fig:evolution} we show the evolution of the different metrics during the
partitioning of the Twitter (left) and the Yahoo! (right) graphs. The Twitter
graph is known for the existence of high-degree hubs \cite{Gonzalez2012}. Due to
these hubs, the Twitter graph can produce highly unbalanced partitions when
partitioned with random partitioning. We will show in Section
\ref{sec:application_performance} how unbalance affects application performance,
and in particular load balancing.

As Figure~\ref{fig:evolution-twitter} shows, the initial maximum normalized
load obtained with random partitioning is high ($\rho = 1.67$), as expected.
However, by applying our approach, the partitioning is quickly balanced ($\rho
= 1.05$), while the ratio of local edges increases steadily. Looking at the
shape of the $score(G)$ curve, notice that initially the global score is
boosted precisely by the increased balance, while after balance is reached,
around iteration 20, it increases following the trend of $\phi$.  Note that we
let the algorithm run for 115 iterations ignoring the halting condition, which
otherwise would have halted the partitioning at iteration 41, as shown by the
vertical black line.

In the case of the Yahoo! graph, the partitioning starts out more balanced than
Twitter. The most notable difference is that the algorithm converges to a good
partitioning with 73\% local edges and a $\rho$ of 1.10 after only 42
iterations. A single iteration on the 1.4B-vertex Yahoo! graph takes on average
200 seconds on an AWS Hadoop cluster consisting of 115 m2.4xlarge instances,
totaling 140 minutes for the entire execution.

\begin{figure*}[t!]
	\centering
		\subfigure[Partitioning of the Twitter graph.]{
			\includegraphics[width=0.47\linewidth]{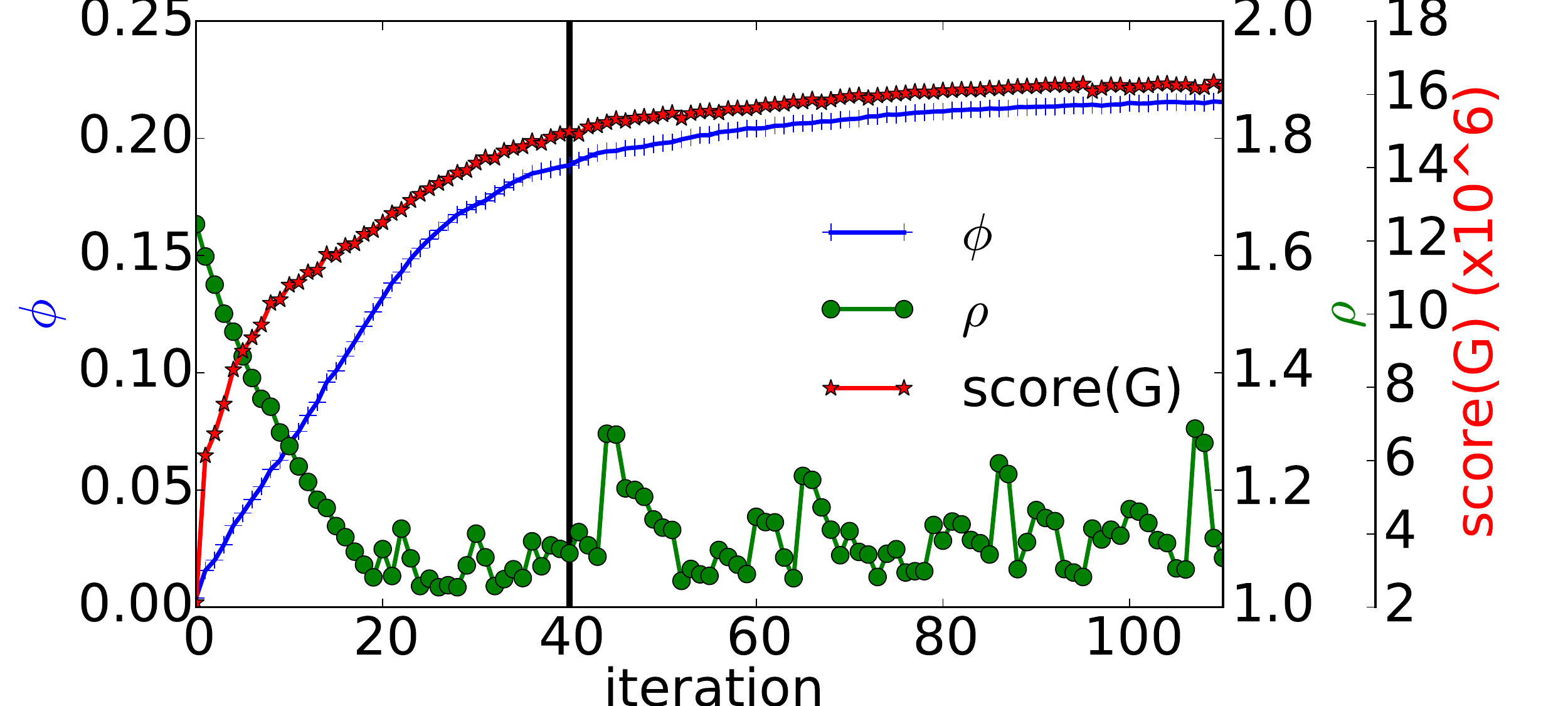}
			\label{fig:evolution-twitter}
		}
		\subfigure[Partitioning of the Yahoo! graph.]{
		 \includegraphics[width=0.47\linewidth]{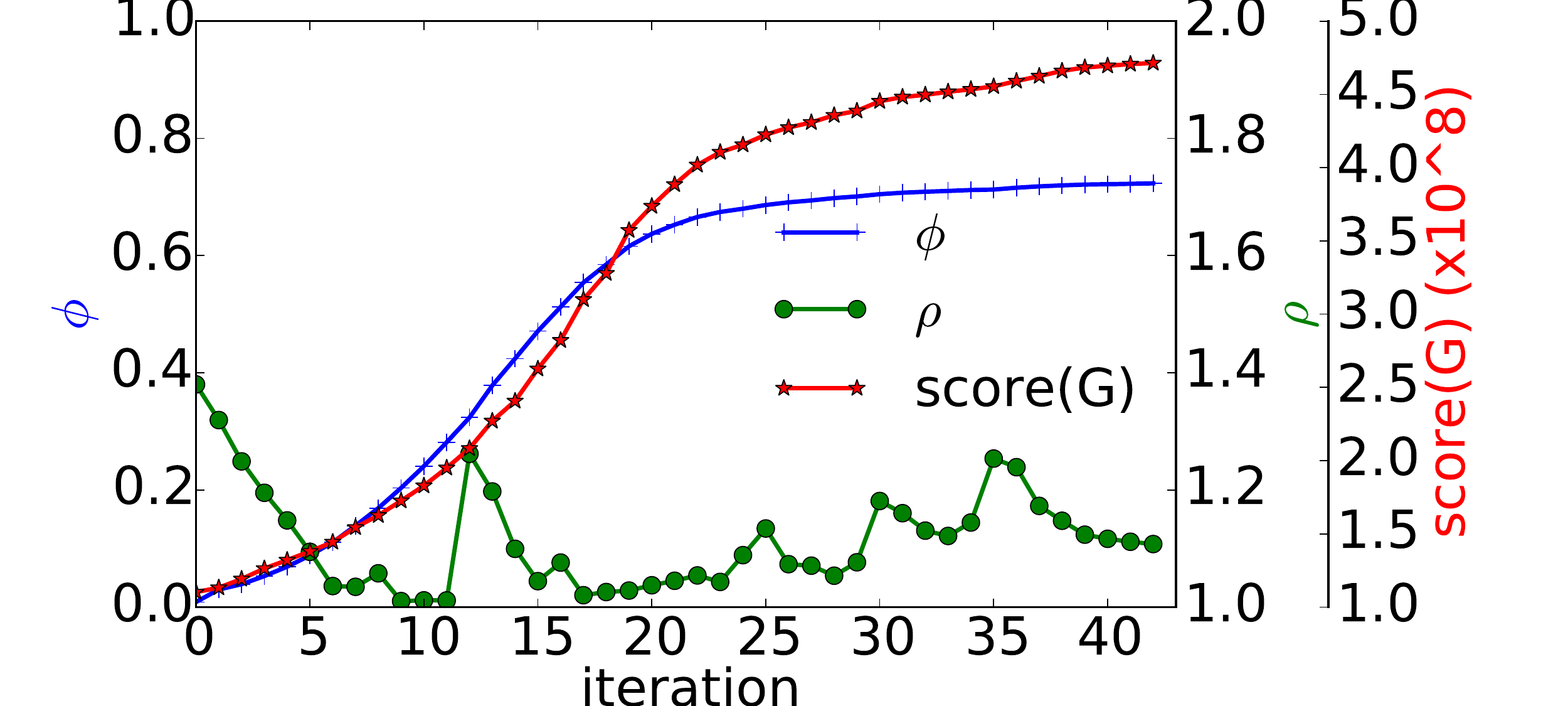}
		 \label{fig:evolution-yahoo}
	   }
	   	\vspace{-10px}
	\caption{Partitioning of (a) the Twitter graph across 256 partitions
	and (b) the Yahoo! web graph across 115 partitions. The figure shows the
	evolution of metrics $\phi$, $\rho$, and $score(G)$ across iterations.
	} 
	\label{fig:evolution}
	\vspace{-10px}
\end{figure*}


\subsubsection{Impact of additional capacity on balance and convergence}
\label{sec:eval:balance}

\begin{figure}[t]
	\centering
	\subfigure[Balance]{
		\includegraphics[width=0.46\linewidth]{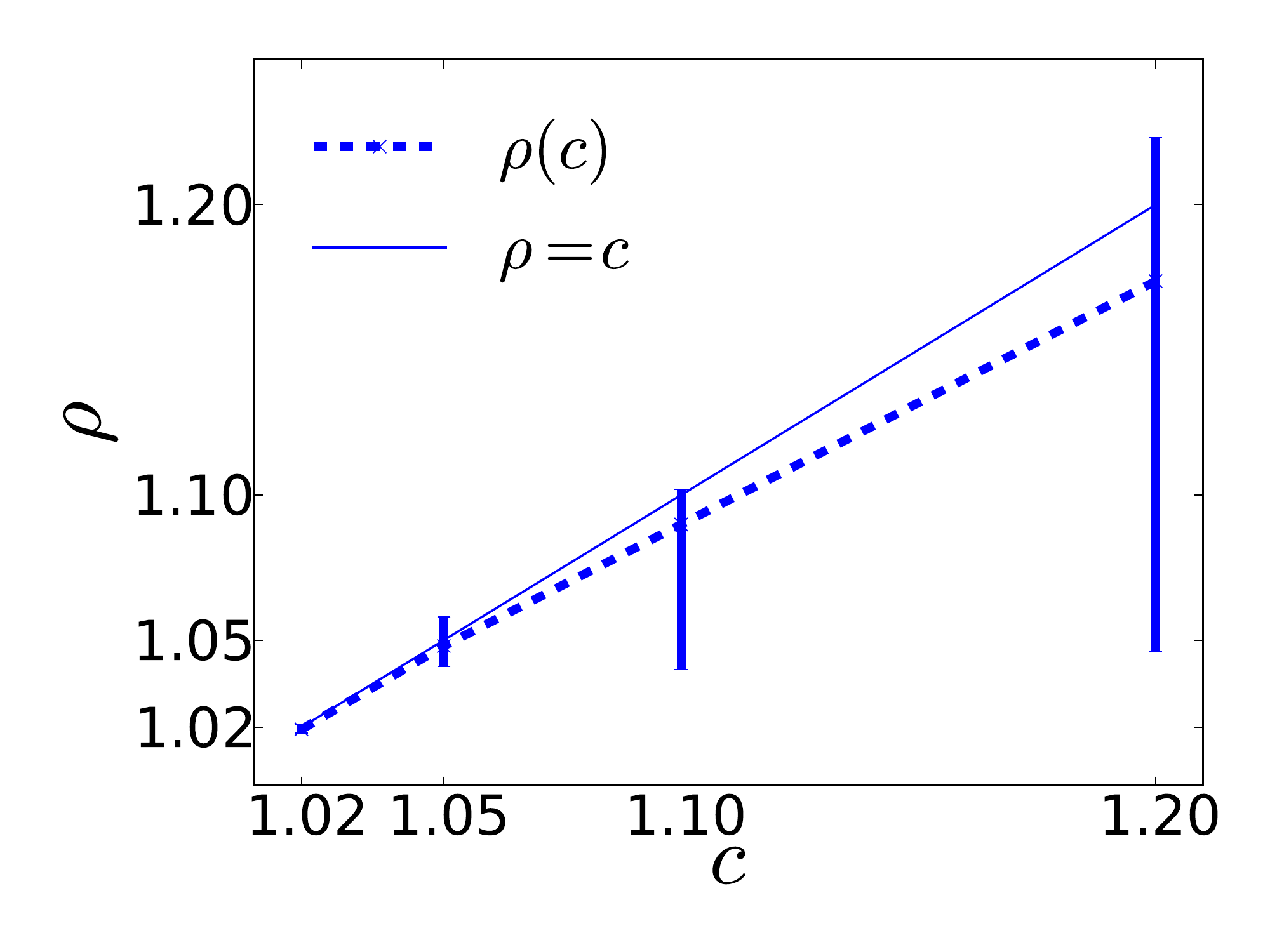}
		\label{fig:c_balance}
	}
	\subfigure[Convergence speed]{
		\includegraphics[width=0.46\linewidth]{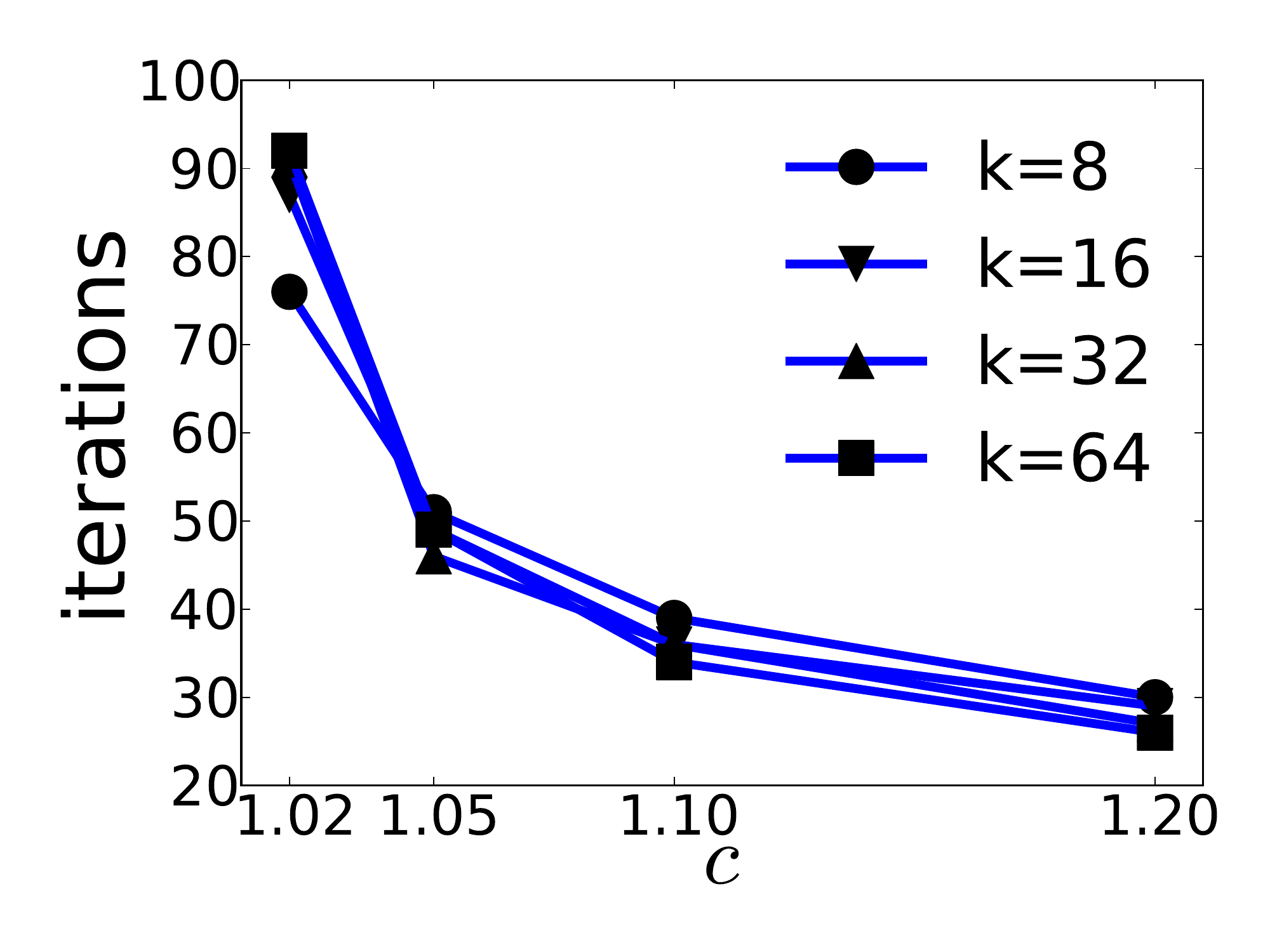}
		\label{fig:c_supersteps}
	}
    \vspace{-10px}
	\caption{Impact of $c$ on the partitioning. Figure (a) shows the relationship between $c$ and $\rho$. Figure (b) shows the relationship between $c$ and speed of convergence.} 
	\vspace{-14px}
\end{figure}

Here, we investigate the effect of parameter $c$ on balance and convergence
speed. Recall that \sys uses parameter $c$ to control the maximum unbalance.
Additionally, parameter $c$ affects convergence speed; larger values of $c$
should increase convergence speed as more migrations are allowed during each
iteration. 

In Section \ref{sec:migrations} we showed that with high probability \sys
respects partition capacities, that is, $maximum~load \leq C$. From the
definitions of $\rho = \frac{maximum~load}{\frac{|E|}{k}}$ and $C = c \cdot
\frac{|E|}{k}$, we derive that with high probability $\rho \leq c$. Therefore,
we can use parameter $c$ to bound the unbalance of partitioning. For instance,
if we allow partitions to store $20\%$ more edges than the ideal value, \sys
should produce a partitioning with a maximum normalized load of up to $1.2$. 

To investigate these hypotheses experimentally, we vary the value of $c$ and
measure the number of iterations needed to converge as well as the final value
of $\rho$. We partition the LiveJournal graph into 8, 16, 32, and 64
partitions, setting $c$ to $1.02$, $1.05$, $1.10$, and $1.20$. We repeat each
experiment 10 times and the average for each value of $c$. As expected, Figure
\ref{fig:c_balance} shows that indeed on average $\rho \leq c$. Moreover, the
error bars show the minimum and maximum value of $\rho$ across the runs. We can
notice that in some cases $\rho$ is much smaller than $c$, and when it is
exceeded, it is exceeded only by a small degree. 

Figure \ref{fig:c_supersteps} shows the relationship between $c$ and the number
of iterations needed to converge. Indeed, a larger value of $c$ speeds up
convergence. These results show how $c$ can be used to control the maximum
normalized load of the partitioning. It is up to the user to decide the
trade-off between balance and speed of convergence.

\subsection{Scalability}

\begin{figure*}[t!]
	\centering
	\subfigure[Runtime vs. graph size]{
		\includegraphics[width=0.29\linewidth]{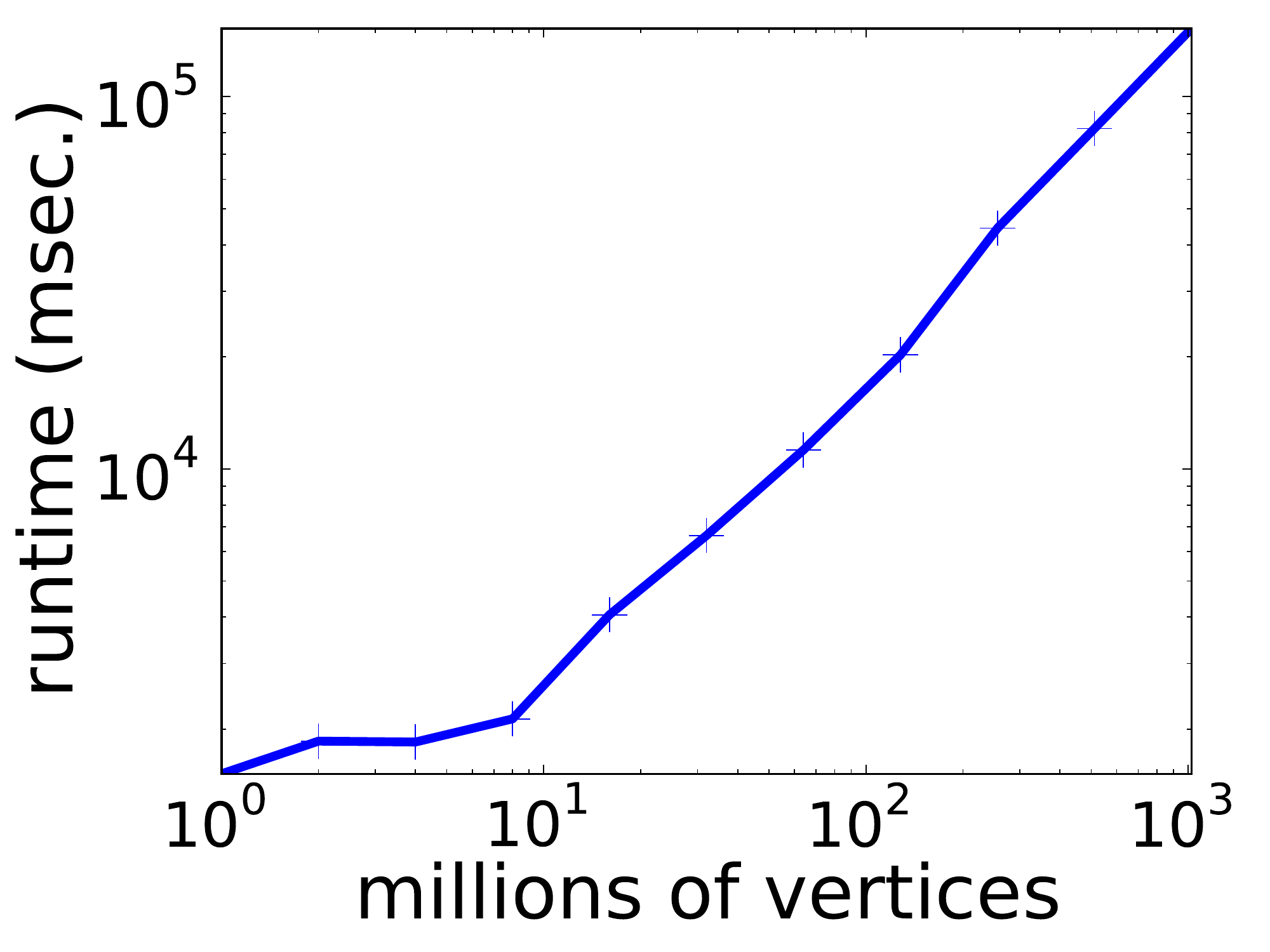}
		\label{fig:scalability_size}
	}
	\hspace{10px}
	\subfigure[Runtime vs. cluster size]{
		\includegraphics[width=0.29\linewidth]{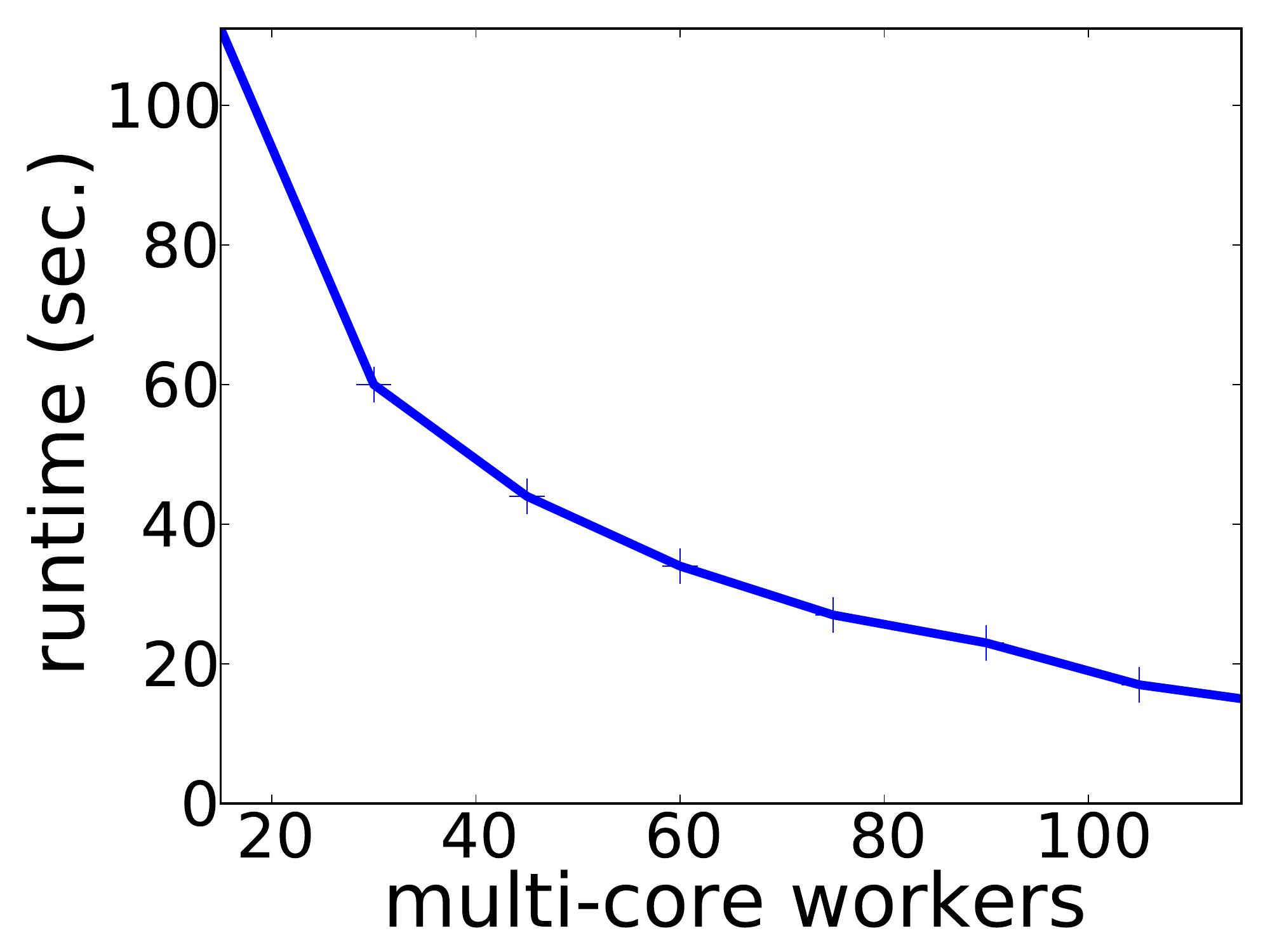}
		\label{fig:scalability_workers}
	}
	\hspace{10px}
	\subfigure[Runtime vs. $k$]{
		\includegraphics[width=0.29\linewidth]{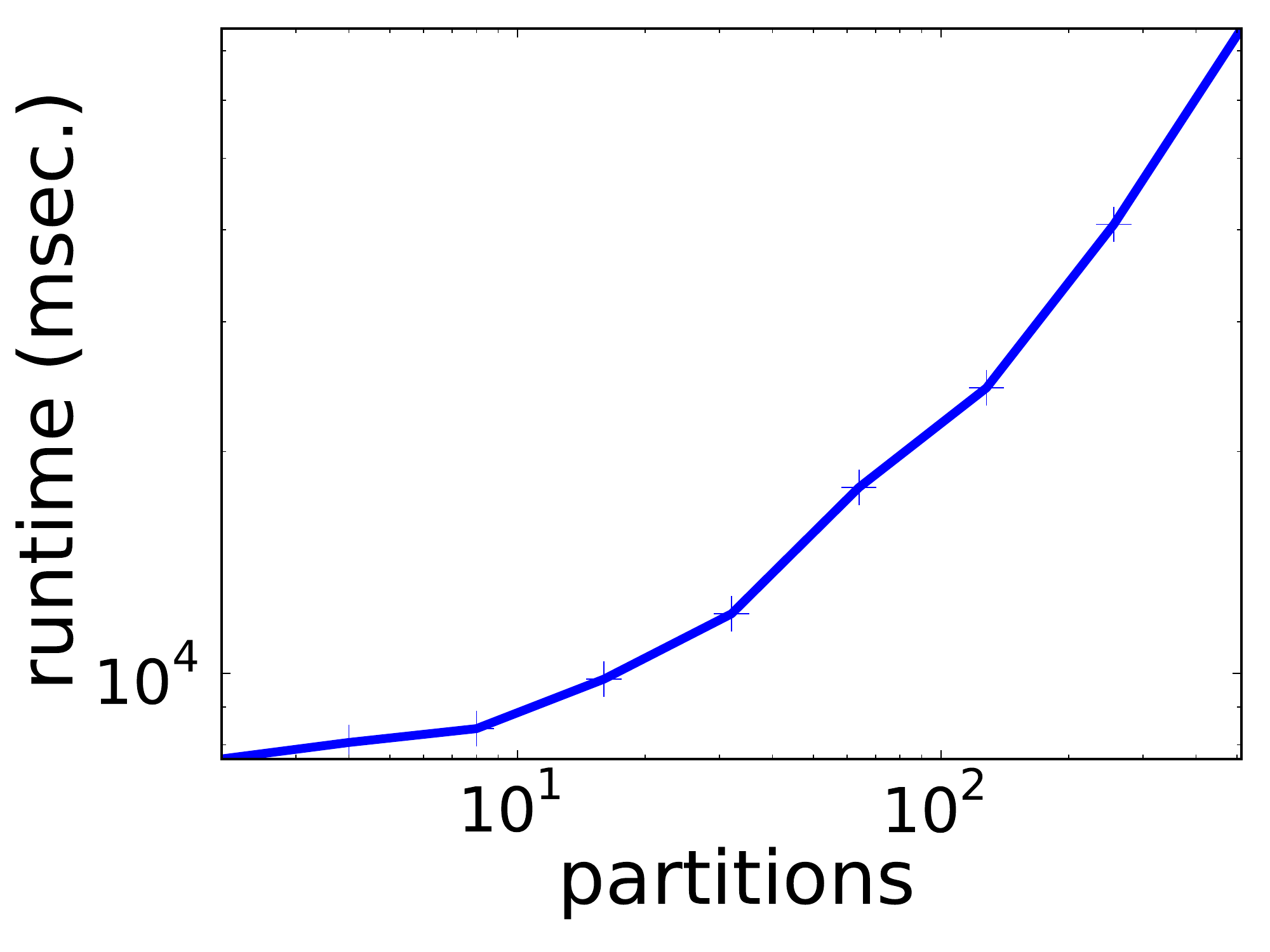}
		\label{fig:scalability_partitions}
	}
	\vspace{-10px}
	\caption{Scalability of \sys. (a) Runtime as a function of the number of vertices, (b) runtime as a function of the number of workers, (c) runtime as a function of the number of partitions.} 
	\label{fig:scalability}
	\vspace{-10px}
\end{figure*}

In these experiments we show that the algorithm affords a scalable
implementation on modern large-scale graph processing frameworks such as Giraph.
To this end, we apply our algorithm on synthetic graphs constructed with the
Watts-Strogatz model \cite{Watts-Colective-1998}. In all these experiments we
set the parameters of the algorithm as described in Section \ref{sec:quality}.
Using synthetic graphs for these experiments allows us to carefully control the
number of vertices and edges, still working with a graph that resembles a
real-world social network or web-graph characterized by ``small-world''
properties. On such a graph, the number of iterations required to partition the
graph does not depend only on the number of vertices, edges and partitions, but
also on its topology, and in particular on graph properties such as clustering
coefficient, diameter etc.

For this reason, to validate the scalability of the algorithm we focus on the
runtime of the \emph{first} iteration, notably the iteration where all vertices
receive notifications from all their neighbors, making it the most deterministic
and expensive iteration. Precisely, we compute the runtime of an iteration as
the sum of the time needed to compute the \emph{ComputeScores} and the following
\emph{ComputeMigrations} supersteps. This approach allows us to factor out the
runtime of algorithm as a function the number of vertices and edges.

Figure \ref{fig:scalability} presents the results of the experiments, executed
on a AWS Hadoop cluster consisting of 116 m2.4xlarge machines. In the first
experiment, presented in Figure \ref{fig:scalability_size}, we focus on the
scalability of the algorithm as a function of the number of vertices and edges
in the graph. For this, we fix the number of outgoing edges per vertex to 40. We
connect the vertices following a ring lattice topology, and re-wire $30\%$ of
the edges randomly as by the function of the $beta$ (0.3) parameter of the
Watts-Strogatz model. We execute each experiment with 115 workers, for an
exponentially increasing number of vertices, precisely from 2 to 1024 million
vertices (or one billion vertices) and we divide each graph in 64 partitions.
The results, presented in a loglog plot, show a linear trend with respect to the
size of the graph. Note that for the first data points the size of the graph is
too small for such a large cluster, and we are actually measuring the overhead
of Giraph.

In the second experiment, presented in Figure \ref{fig:scalability_workers}, we
focus on the scalability of the algorithm as a function of the number of
workers. Here, we fix the number of vertices to 1 billion, still constructed as
described above, but we vary the number of workers linearly from 15 to 115 with
steps of 15 workers (except for the last step where we add 10 workers).
The drop from 111 to 15 seconds with 7.6 times more workers represents
a speedup of 7.6.

In the third experiment, presented in Figure \ref{fig:scalability_workers}, we
focus on the scalability of the algorithm as a function of the number of
partitions. Again, we use 115 workers and we fix the number of vertices to 1
billion and construct the graph as described above. This time, we increase the
number of partitions exponentially from 2 to 512. Also here, the loglog plot
shows a near-linear trend, as the complexity of the heuristic executed by each
vertex is proportional to the number of partitions $k$, and so is cost of
maintaining partition loads and counters through the sharded aggregators
provided by Giraph.

\subsection{Partitioning dynamic graphs}\label{sec:eval:dynamic}
Due to the dynamic nature of graphs, the quality of an initial partitioning
degrades over time. Re-partitioning from scratch can be an expensive task if
performed frequently and with potentially limited resources.  In this section,
we show that our algorithm minimizes the cost of adapting the partitioning to
the changes, making the maintenance of a well-partitioned graph an affordable
task in terms of time and compute resources required.

\begin{figure}[t]
	\centering
	\subfigure[Cost savings]{
		\includegraphics[width=0.49\linewidth]{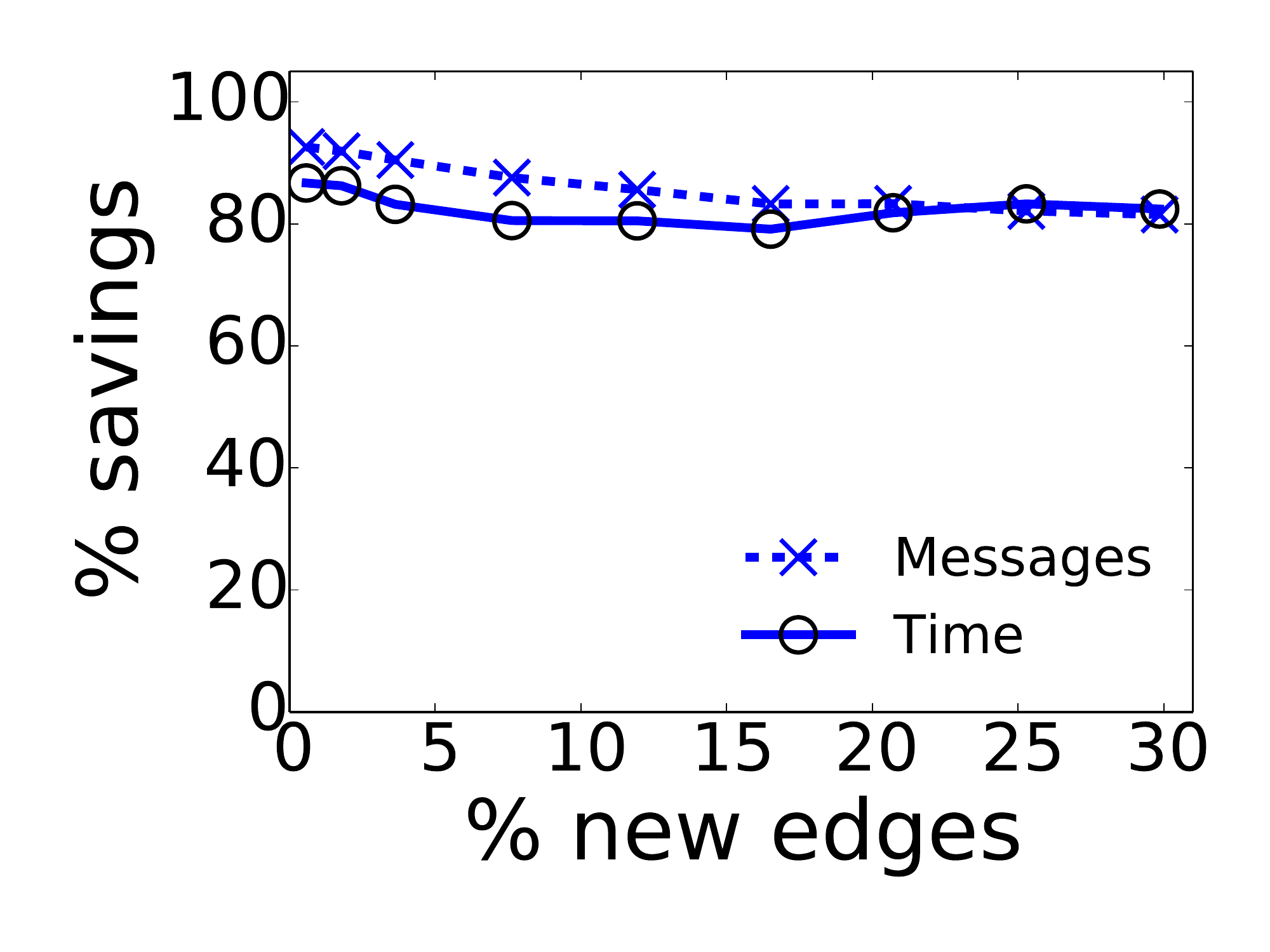}
		\label{fig:time_vs_edges}
	}
	\hspace{-16px}
	\subfigure[Partitioning stability]{
		\includegraphics[width=0.49\linewidth]{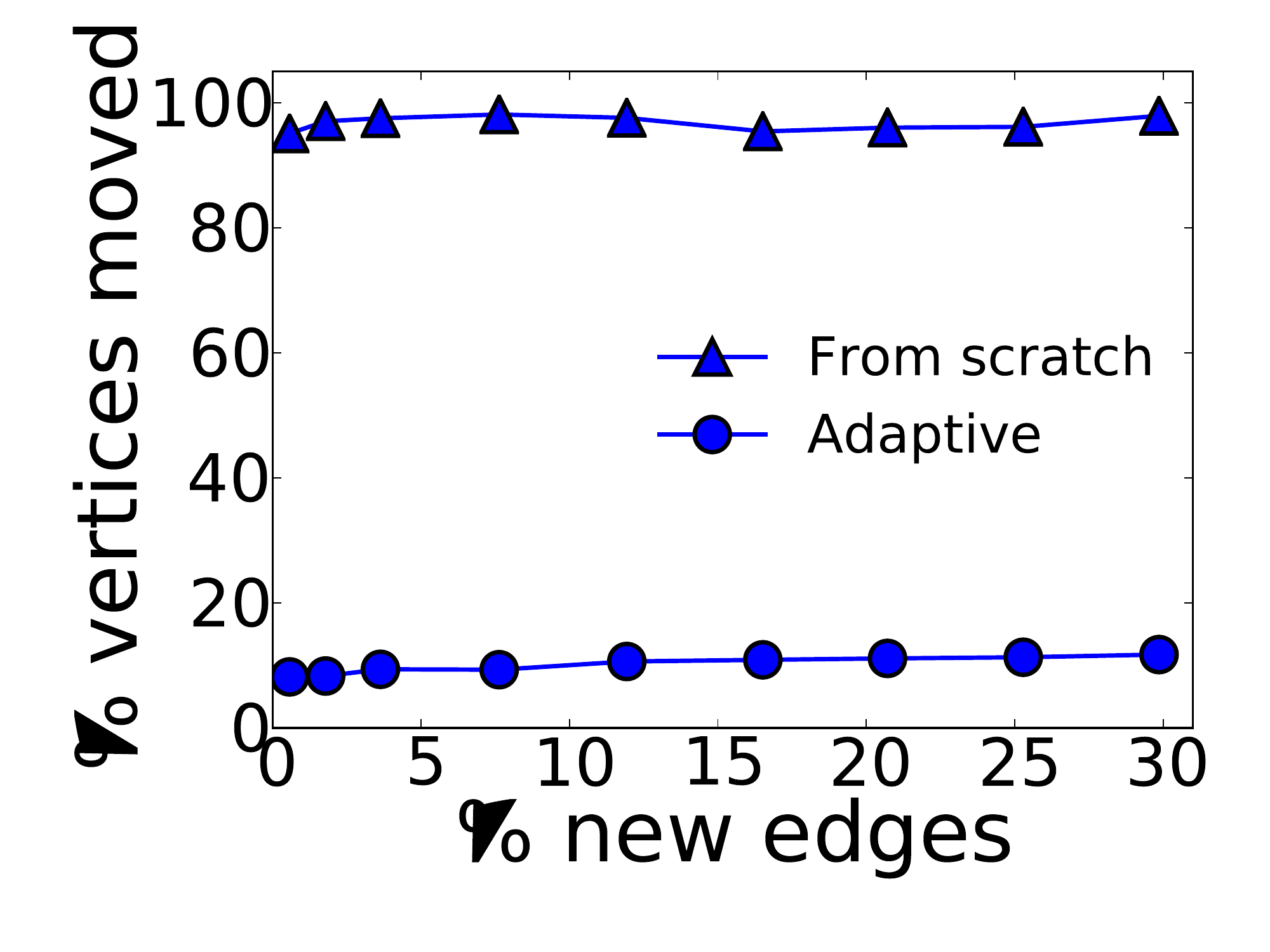}
		\label{fig:stability_vs_edges}
	}
    \vspace{-10px}
	\caption{Adapting to dynamic graph changes. We vary the percentage of
	new edges in the graph and compare our adaptive re-partitioning
	approach and re-partitioning from scratch with respect to (a) the
	savings in processing time and messages exchanged, and (b) the fraction
	of vertices that have to move upon re-partitioning. } 
	\vspace{-10px}
\end{figure}

Specifically, we measure the savings in processing time and number of messages
exchanged (i.e. load imposed on the network) relative to the approach of
re-partitioning the graph from scratch. We track how these metrics vary as a
function of the degree of change in the graph. Intuitively, larger graph
changes require more time to adapt to an optimal partitioning.

For this experiment, we take a snapshot of the Tuenti~\cite{tuenti} social
graph that consists of approximately 10 million vertices and 530 million edges,
and perform an initial partitioning.  Subsequently, we add a varying number of
edges that correspond to actual new friendships and measure the above metrics.
We perform this experiment on an AWS Hadoop cluster consisting of 10 m2.2xlarge
instances.

Figure~\ref{fig:time_vs_edges} shows that for changes up to 0.5\%, our approach
saves up to 86\% of the processing time and, by reducing vertex migrations, up
to 92\% of the network traffic. Even for large graph changes, the algorithm
still saves up to 80\% of the processing time. Note that in every case our
approach converges to a balanced partitioning, with a maximum normalized load
of approximately $1.047$, with 67\%-69\% local edges, similar to a
re-partitioning from scratch.

\subsection{Partitioning stability}

Adapting the partitioning helps maintain good locality as the graph changes,
but may also require the graph management system (e.g. a graph DB) to move
vertices and their associated state (e.g. user profiles in a social network)
across partitions, potentially impacting performance. Aside from efficiency,
the value of an adaptive algorithm lies also in maintaining \emph{stable}
partitions, that is, requiring only few vertices to move to new partitions
upon graph changes.  Here, we show that our approach achieves this goal. 


We quantify the stability of the algorithm with a metric we call
\emph{partitioning difference}. The partitioning difference between two
partitions is the percentage of vertices that belong to different partitions
across two partitionings. This number represents the fraction of vertices
that have to move to new partitions.  Note that this metric is not the same as
the total number of migrations that occur during the execution of the algorithm
which only regards cost of the execution of the algorithm per se.

In Figure~\ref{fig:stability_vs_edges}, we measure the resulting partitioning
difference when adapting and when re-partitioning from scratch as a function of
the percentage of new edges.  As expected, the percentage of vertices that have
to move increases as we make more changes to the graph. However, our adaptive
approach requires only 8\%-11\% of the vertices to move compared to a 95\%-98\%
when re-partitioning, minimizing the impact.

\subsection{Adapting to resource changes}

\begin{figure}[t]
	\centering
	\subfigure[Cost savings]{
		\includegraphics[width=0.49\linewidth]{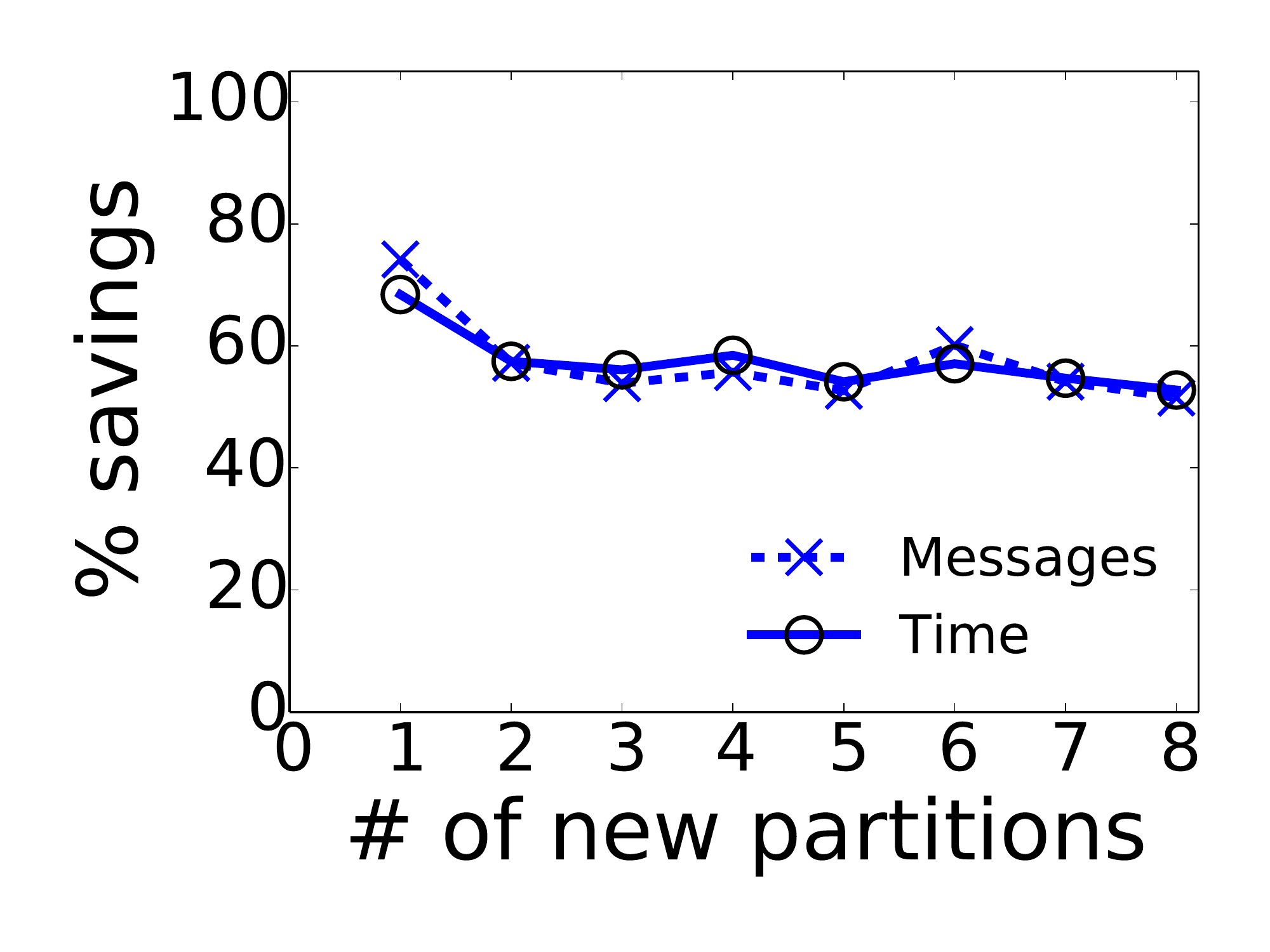}
		\label{fig:time_vs_newpartitions}
	}
	\hspace{-16px}
	\subfigure[Partitioning stability]{
		\includegraphics[width=0.49\linewidth]{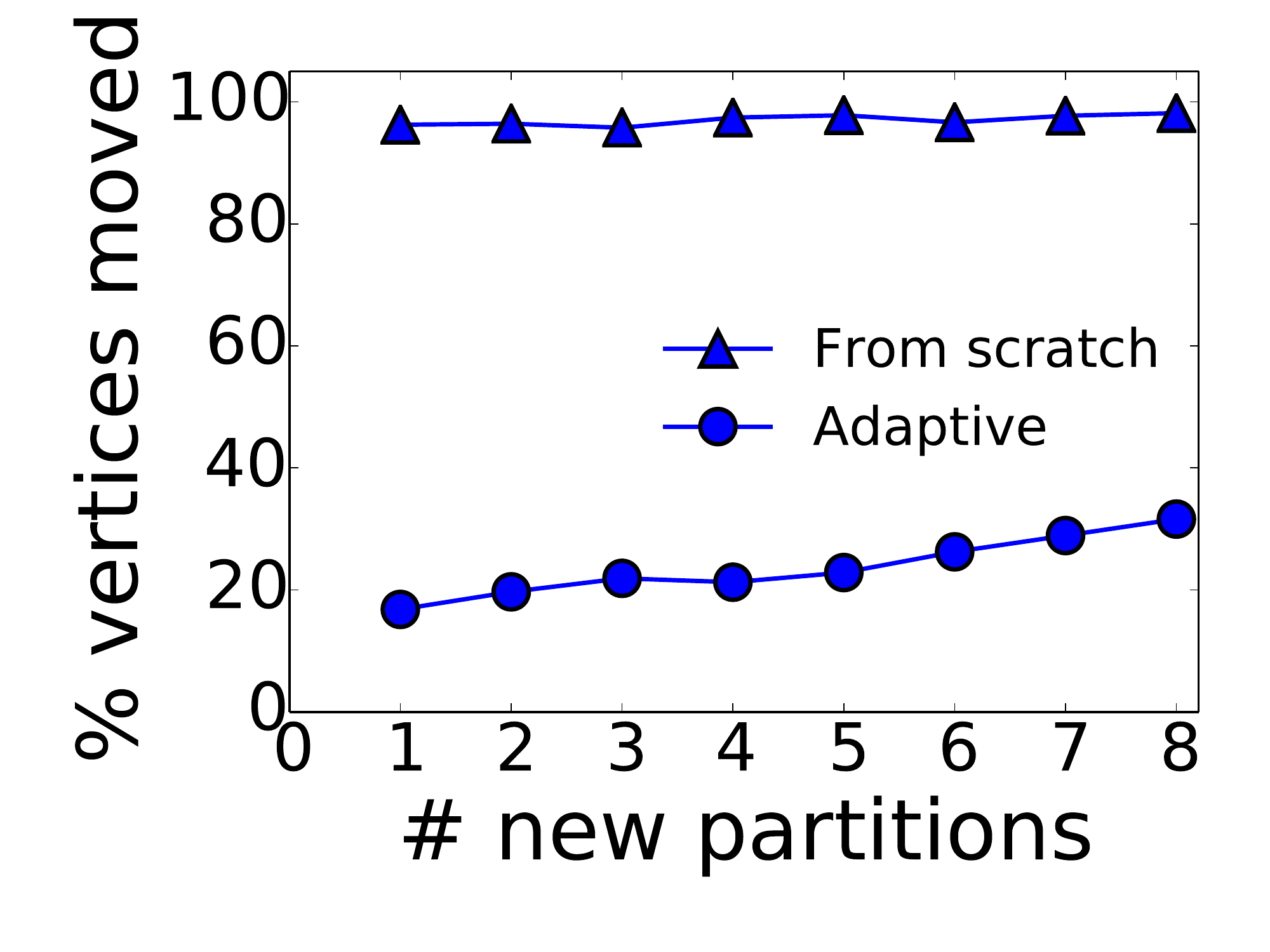}
		\label{fig:stability_vs_newpartitions}
	}
		\vspace{-10px}
	\caption{Adapting to resource changes.  We vary the number of new
	partitions and compare our adaptive approach and re-partitioning from
	scratch with respect to (a) the savings in processing time and messages
	exchanged, and (b) the fraction of vertices that have to move upon
	re-partitioning.  } 
	\label{fig:adaptresources}
	\vspace{-10px}
\end{figure}

Here, we show that \sys can efficiently adapt the partitioning when resource
changes force a change in the number of partitions. Initially, we partition the
Tuenti graph snapshot described in Section~\ref{sec:eval:dynamic} into 32
partitions.  Subsequently we add a varying number of partitions and either
re-partition the graph from scratch or adapt the partitioning with \sys.

Figure~\ref{fig:time_vs_newpartitions} shows the savings in processing time and
number of messages exchanged as a function of the number of new partitions.  As
expected, a larger number of new partitions requires more work to converge to a
good partitioning.  When increasing the capacity of the system by only 1
partition, \sys adapts the partitions 74\% faster relative to a
re-partitioning. 

Similarly to graph changes, a change in the capacity of the compute system may
result in shuffling the graph. In Figure~\ref{fig:stability_vs_newpartitions},
we see that a change in the number of partitions can impact partitioning
stability more than a large change in the input graph
(Figure~\ref{fig:stability_vs_edges}).  Still, when adding only 1 partition
\sys forces less than 17\% of the vertices to shuffle compared to a 96\% when
re-partitioning from scratch. The high percentage when re-partitioning from
scratch is expected due to the randomized nature of our algorithm.  Note,
though, that even a deterministic algorithm, like modulo hash partitioning, may
suffer from the same problem when the number of partitions changes.

\subsection{Impact on application performance}
\label{sec:application_performance}
The partitioning computed by \sys can be used by different graph management
systems, to improve their performance.  In this section, we use \sys to
optimize the execution of the Giraph graph processing system itself.  After
partitioning the input graph with \sys, we instruct Giraph to use the computed
partitioning and run real analytical applications on top. We then measure the
impact on performance compared to using standard hash partitioning.  

We use our computed partitioning in Giraph as follows. The output of \sys is a
list of pairs $(v_{i}, l_{j})$ that assigns each vertex to a partition. We use
this output to ensure that Giraph places vertices assigned to the same
partition on the same physical machine worker.  By default, when Giraph loads a
graph for computation, it assigns vertices to workers according to hash
partitioning, i.e. vertex $v_{i}$ is assigned to one of the $k$ workers
according to $h(v_{i})~\emph{mod}~k$.  We define a new vertex id type
$v'_i = (v_{i}, l_{j})$ that encapsulates the computed partition as well.
We then plug a hash function that uses only the $l_{j}$ field of the pair,
ensuring that vertices with the same label are placed on the same worker.


First, we assess the impact of partitioning balance on the actual load balance
of the Giraph workers.  In a synchronous processing engine like Giraph, an
unbalanced partitioning results in the less loaded workers idling at the
synchronization barrier.  To validate this hypothesis, we partition the Twitter
graph across 256 partitions and run 20 iterations of the PageRank algorithm on
a cluster with 256 workers using (i) standard hash partitioning (random), and
(ii) the partitioning computed by \sys.  For
each run, we measure the time to compute a superstep by all the workers (Mean),
the fastest (Min) and the slowest (Max), and compute the standard deviation
across the 20 iterations. Table \ref{tab:balance_runtime} shows the results.

The results indicate that with hash partitioning the workers are idling on
average for $31\%$ of the superstep, while with Spinner for only $19\%$. While
the shorter time needed to compute a superstep can be imputed to the diminished
number of cut edges, the decreased idling time is an effect of a more even load
spread across the workers. 

\begin{table}
	\centering
	\begin{tabular}{l r r r r}
		\toprule
		Approach & Mean & Max. & Min. \\
		\hline
		Random & $5.8s \pm 2.3s$ &  $8.4s \pm 2.1s$ & $3.4s \pm 1.9s$ \\
		\textbf{\sys} & $4.7s \pm 1.5s$ & $5.8s \pm 1.3s$ &  $3.1s \pm 1.1s$ \\
       \bottomrule
	\end{tabular}
	\caption{Impact of partitioning balance on worker load. The table shows the time spent by 
	workers to conclude a superstep.}
	\label{tab:balance_runtime}
	\vspace{-10px}
\end{table}

Second, we measure the impact of the partitioning on processing time. We
partitioned three graphs with Spinner and hash partitioning, and we compared
the time to run three representative algorithms commonly used inside graph
analytical pipelines. Shortest Paths, computed through BFS is commonly used to
study the connectivity of the vertices and centrality, PageRank is commonly
used at the core of ranking graph algorithms, and Connected Components, as a
general approach to finding communities. We present the results in Figure
\ref{fig:app_perf}.

Notice that using the partitionings computed by \sys we significantly improve the
performance across all graphs and applications. In the case of the Twitter
graph, which is denser and harder to partition, the improvement ranges from
$25\%$ for SP to $35\%$ for PR. In the case of LiveJournal and Tuenti, the
running time decreases by up to 50\%.

\begin{figure}[t]
	\centering
		\includegraphics[width=0.94\linewidth]{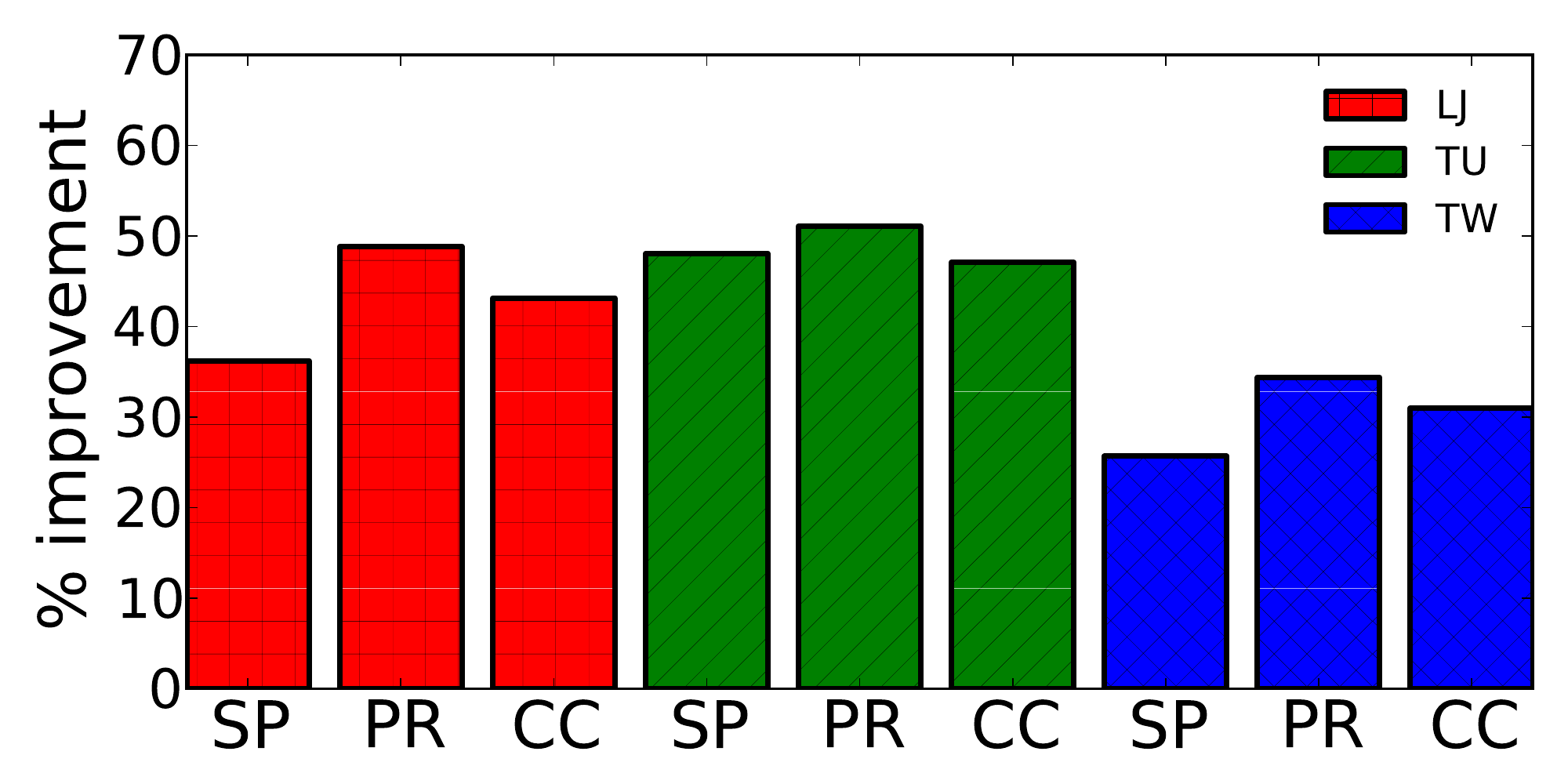}
			\vspace{-10px}
	\caption{Impact on application performance. We measured the runtime 
	improvement between hash partitioning and Spinner. We evaluated three different 
	applications, Single Source Shortest Paths/BFS (SP), PageRank (PR), and Weakly 
	Connected Components (CC). The LJ graph was evaluated across 16 partitions, TU across
	32 partitions and TW across 64.} 
	\label{fig:app_perf}
	\vspace{-6px}
\end{figure}

%% file: related.tex
Graph partitioning that has been studied in various domains.  In this section,
we present the related work on k-way balanced graph partitioning.



METIS~\cite{Karypis1995} is an offline partitioning algorithm and is considered
the golden standard against new approaches. It is known to produce partitions
with very good locality and balance. However, it is not suitable for very large
graphs, due to high computational and space complexity, and the need to
repartition graphs from scratch every time a change is introduced in the graph
or the number of partitions.

Streaming algorithms~\cite{Stanton2012, Tsourakakis2014, Nishimura2013} avoid
this complexity through lightweight heuristics that assign vertices to
partitions in only one pass over the graph. However, parallelizing these
algorithms requires every participating machine to maintain a consistent view
of the partitioning assignments of the entire graph. This requires distribution
of the assignments across the cluster and coordination among the machines.
Implementing this mechanism in a scalable and efficient manner is challenging.
To the best of our knowledge, none of these approaches have been implemented
on top of scalable data processing models.

The closest approaches to \sys are \cite{Ugander2013, Wang2013}. The former
applies LPA to the MapReduce model, by attempting to improve locality through
iterative vertex migrations across partitions. However, to guarantee balanced
partitions, it executes a centralized linear solver between any two iterations.
The complexity of lineary system is quadratic to the number of partitions and
proportional to the size of the graph, making it expensive for large graphs.
Moreover, MapReduce is known to be inefficient for iterative computations.  The
approach in \cite{Wang2013} computes a k-way vertex-based balanced partitioning.
It uses LPA to coarsen the input graph and then applies Metis to the coarsened
graph. At the end, it projects the Metis partitioning back to the original
graph. While the algorithm is scalable, we have found that for large number of
partitions and skewed graphs, the locality it produces is lower than \sys due to
the coarsening. We also found that the approach is very sensitive to its two
parameters for whom no intuition is available (differently from \sys that
requires only one parameter for which we provide a rationale). Further, the
approach is designed to run on the Trinity engine and is not suitable for
implementation on a synchronous model such as Pregel. None of the two solutions
investigates how to adapt a partitioning upon changes in the graph or the
compute environment.




Mizan~\cite{Khayyat2013} views the problem of dynamically adapting the graph
partitioning from a different perspective.  It monitors runtime performance of
the graph processing system, for instance, to find hot\-spots in specific
machines, and migrates vertices across workers to balance the load during the
computation of analytical applications. Their solution is specific to a graph
processing system and orthogonal to the problem of k-way balanced partitions. 
SPAR~\cite{Pujol2011} and Sedge~\cite{Yang2012} also consider the problem of
adapting graph distribution. However, they focus on minimizing latency for
simple graph queries, and address the problem through replication of the
vertices across machines.

%% file: concl.tex
We presented \emph{\sys}, a scalable and adaptive graph partitioning algorithm
built on the Pregel abstraction. By sacrificing strict guarantees on balance,
\sys is practical for large-scale graph management systems.  Through an
extensive evaluation on a variety of graphs, we showed that \sys computes
partitions with locality and balance comparable to the state-of-the-art, but
can do so at a scale of at least billion-vertex graphs.  At the same time, its
support for dynamic changes makes it more suitable for integrating into real
graph systems. These properties makes \sys possible to use as a generic
replacement of the de-facto standard, hash partitioning, in cloud systems.
Toward this, our scalable, open source implementation on Giraph makes \sys easy
to use on any commodity cluster.



%% file: appendix.tex
\newtheorem*{prop}{Proposition~\ref{prop:bconnectivity}}
\begin{prop}
	If the partition graph $\{P_{t>0}\}$ is $B$-connected, one can always find constants $\mu \in (0,1)$ and $q \in \mathbf{R}^{+}$, for which Spinner converges exponentially $\|x_t - x^{\star} \|_{\infty} / \|x_0\|_{\infty} \leq q \mu^{t-1}$ to an even balancing $x^{\star} = [C\ C\ \ldots\ C]^\top\hspace{-1mm}$, with $C = |E|/k$. 
\end{prop}
\begin{proof}
Even though $0 \leq [X_t]_{ij} \leq 1$ is time-dependent and generally unknown, it has the following properties:
\begin{itemize}
\item $X_t$ is \emph{1-local}: $[X_t]_{ij} > 0$ iff $(l_i,l_j) \in Q_t$ and $[X_t]_{ii} > 0$ for all $i$
\item $X_t$ is \emph{row-stochastic}: $\sum_{j = 1}^{k} [X_t]_{ij} = 1$ for all $i$
\item $X_t$ is \emph{uniformly bounded}: if $[X_t]_{ij} > 0$ then $[X_t]_{ij} \geq 1/n$
\end{itemize}
In general, the above properties do not guarantee ergodicity. Nevertheless, we can obtain stronger results if the partition graph is $B$-connected.
From~\cite[Lemma~5.2.1]{tsitsiklis1984problems} (see also~\cite[Theorem~2.4]{touri2012}), 
$X_{t:1}$ is ergodic and 
	$x_\infty = \lim_{t\rightarrow\infty} X_{t:1}\, x_0 = \mathbf{1} \pi^\top x_0,$
where $\mathbf{1}$ is the one vector and $\pi$ is a stochastic vector. Furthermore, constants $\mu \in (0,1)$ and $q \in \mathbf{R}^{+}$ exist for which 
\begin{align*}
	\frac{\| x_t - x_\infty \|_{\infty}}{ \|x_0\|_{\infty}} \leq \| X_{t:1} - \mathbf{1} \pi^\top \|_{\infty} \leq q \mu^{t-1}.
\end{align*}
Furthermore, from ergodicity, $\lim_{t\rightarrow\infty} X_{t:1}$ is a rank-one matrix and $x_\infty({l_i}) = x_\infty(l_j)$ for all partitions $l_i,l_j$ and by construction the total load is always equal to the number of graph edges. It follows that, for all $l$, $x_\infty(l) = \sfrac{|E|}{k} = C = x^\star(l)$.
\end{proof}

\newtheorem*{prop2}{Proposition~\ref{prop:general}}
\begin{prop2}
	Spinner converges in bounded time.
\end{prop2}
\begin{IEEEproof}
Convergence is proven by first splitting $\{P_{t>0}\}$ according to Lemma~\ref{lemma:partitioning} and then applying Proposition~\ref{prop:bconnectivity} for each of the resulting subgraphs. The time required until $ \|x_{t} - x^{\star} \|_{\infty} / \| x_0\|_{\infty} \leq \varepsilon$ is at most $\leq \log_{\mu}{\left(\varepsilon / q\right)} + 1 + T$, which is equal to the splitting time $T$ added to the time for exponential convergence. 
\end{IEEEproof}

\begin{lemma}
	Labels $l \in L$ can be always split into $p$ subsets $L_1, L_2, \ldots, L_p$ having the following properties:
	\begin{enumerate}
		\item Subsets $L_i$ are non-empty, disjoint, and cover $L$.
		\item Any induced subgraph $\{P_{t>0}(L_i)\}$, i.e., the subgraph of $\{P_{t>0}\}$ which includes only nodes in $L_i$ and all edges with both endpoints in $L_i$, is $B$-connected.
		\item A bounded time $T$ exists after which $\{P_{t>T}(L_i)\}$ and $\{P_{t>T}(L_j)\}$ are disconnected $\forall i,j$. In other words, no edge connects two subgraphs for $t>T$. 
	\end{enumerate}
	\label{lemma:partitioning}
\end{lemma}
\begin{IEEEproof}
To begin with, observe that when $\{P_{t>0}\}$ is $B$-connected, the statement is obvious for $L_1 = L$ and $p=1$. 
Assume that $B$-connectivity does not hold and let $T_1$ be the -\emph{latest}- time for which $\{P_{T_1 \geq t>0}\}$ is $B$-connected, under all possible bounded $B$. Notice that, by construction, a $p_1$-split exists which satisfies the \emph{first} and \emph{third} property (each subset includes all labels that exchange load after $T_1$). If such a partitioning did not exist, then $\{P_{T_1+1 \geq t > 0}\}$ must have been $B$-connected---a contradiction. 
Though this first splitting does not guarantee the \emph{second} property, \ie that $\{P_{t>0}(L_i)\}$ are $B$-connected, we can find the correct subsets by recursively re-splitting each subgraph in the same manner. This recursive re-splitting will always terminate because: (i) all subsets are non-empty and (ii) the trivial $k$-splitting $L_i = \{ l_i \}$ is always $B$-connected. Hence time $T$, chosen as the time of the last splitting, is bounded. 
\end{IEEEproof}

\newtheorem*{prop3}{Proposition~\ref{prop:tail_inequality}}
\begin{prop3}
	The probability that at iteration $i+1$ the load $b_{i+1}(l)$ exceeds the capacity by a factor of $\epsilon\, r_i(l)$ is 
\begin{align}
	\mathbf{Pr}(b_{i+1}(l) \geq C + \epsilon \, r_i(l)) &\leq e^{- 2\, |M(l)|\, \Phi(\epsilon)}, 	
\end{align}
where $\Phi(\epsilon) = \left(\frac{\epsilon \, r_i(l)}{\Delta - \delta}\right)^{\hspace*{-1mm}2}$ and $\delta$, $\Delta$ is the minimum and maximum degree of the vertices in $M(l)$, respectively. 
\end{prop3}

\begin{IEEEproof}
Let $X_{v} $ be a random variable which becomes $0$ when vertex $v$ does not
migrate and $deg(v)$ otherwise.  The expected value of $X_{v}$ is
\begin{align*}
 	\mathbf{E}(X_v) = 0 \cdot (1 - p) + deg(v)\, p = deg(v) \, p.
\end{align*}
The total load carried by the vertices that migrate is described by the random
variable $X = \sum_{v \in M(l)} X_{v}$ and has expected value
\begin{align*}
	\mathbf{E}(X) = \mathbf{E}\left(\sum_{v \in M(l)} \hspace{-2mm}X_v\right) = \sum_{v \in M(l)} \hspace{-2mm} \mathbf{E}(X_v) = p \hspace{-1.5mm} \sum_{v \in M(l)} \hspace{-2.0mm} deg(v) = r(l).
\end{align*}
We want to bound the probability that $X$ is larger than
$r(l)$, that is, the number of edges that migrate to $l$ exceeds the remaining
capacity of $l$.  Using Hoeffding's method, we have that for any $t > 0$,
\begin{align*}
	\mathbf{Pr}(X - \mathbf{E}(X) \geq t) &\leq \text{exp}\hspace*{-0.5mm}\left({-\frac{2 |M(l)|^2 t^2}{\sum\limits_{v \in M(l)} \hspace*{-2mm}(\Delta - \delta)^2}}\right) = \text{exp}\hspace*{-0.5mm}\left({-\frac{2 |M(l)| t^2}{(\Delta - \delta)^2}}\right),
\end{align*}
where $\delta$ and $\Delta$ are the minimum and maximum degree of the vertices in $M(l)$, respectively. Setting $t = \epsilon\,\mathbf{E}(X)$, we obtain the desired upper bound:
\begin{align*}
	\mathbf{Pr}(X \geq (1+\epsilon)\mathbf{E}(X)) &= \mathbf{Pr}(X + b(l) \geq C + \epsilon r(l) ) \nonumber\\
		&\leq \text{exp}\hspace*{-0.5mm}\left({- 2 \,|M(l)|\left(\frac{\epsilon \, r(l)}{\Delta - \delta}\right)^{\hspace*{-1mm}2}}\right)
\end{align*}
\end{IEEEproof}
